\documentclass[a4paper,11pt,oneside]{article}

\usepackage[english]{babel}
\usepackage[T1]{fontenc}
\usepackage[utf8]{inputenc}

\usepackage[pdftex,unicode]{hyperref}
\hypersetup{pdftitle=Estimation of Ornstein-Uhlenbeck Process Using Ultra-High-Frequency Data with Application to Intraday Pairs Trading Strategy}
\hypersetup{pdfauthor=Vladimír Holý and Petra Tomanová}

\usepackage[margin=60pt]{geometry}
\usepackage{amsmath}
\usepackage{amssymb}
\usepackage{amsthm}
\newtheorem{proposition}{\sc Proposition}

\usepackage{color}
\definecolor{mycolor}{rgb}{0.725,0.345,0.208}
\hypersetup{colorlinks=true, linkcolor=mycolor, anchorcolor=mycolor, citecolor=mycolor, filecolor=mycolor, urlcolor=mycolor}

\usepackage[group-minimum-digits=3]{siunitx}

\usepackage{graphicx}
\usepackage{float}
\usepackage{hhline}
\usepackage{multirow}
\usepackage{booktabs}

\usepackage[authoryear]{natbib}

\begin{document}

\begin{center}
{\Large \bfseries Estimation of Ornstein--Uhlenbeck Process Using Ultra-High-Frequency Data with Application to Intraday Pairs Trading Strategy}
\end{center}

\begin{center}
{\bfseries Vladimír Holý} \\
Prague University of Economics and Business \\
Winston Churchill Square 4, 130 67 Prague 3, Czech Republic \\
\href{mailto:vladimir.holy@vse.cz}{vladimir.holy@vse.cz} \\
\end{center}

\begin{center}
{\bfseries Petra Tomanová} \\
Prague University of Economics and Business \\
Winston Churchill Square 4, 130 67 Prague 3, Czech Republic \\
\href{mailto:petra.tomanova@vse.cz}{petra.tomanova@vse.cz}
\end{center}

\noindent
\textbf{Abstract:}
When stock prices are observed at high frequencies, more information can be utilized in estimation of parameters of the price process. However, high-frequency data are contaminated by the market microstructure noise which causes significant bias in parameter estimation when not taken into account. We propose an estimator of the Ornstein--Uhlenbeck process based on the maximum likelihood which is robust to the noise and utilizes irregularly spaced data. We also show that the Ornstein--Uhlenbeck process contaminated by the independent Gaussian white noise and observed at discrete equidistant times follows an ARMA(1,1) process. To illustrate benefits of the proposed noise-robust approach, we introduce a novel intraday pairs trading strategy based on the mean-variance optimization. In an empirical study of 7 Big Oil companies, we show that the use of the proposed estimator of the Ornstein--Uhlenbeck process leads to an increase in profitability of the pairs trading strategy.
\\

\\

\\

\section{Introduction}
\label{sec:intro}

In finance, many different time series tend to move to their mean values over time. This behavior is known as the mean reversion and is often captured by the Ornstein--Uhlenbeck process \citep{Uhlenbeck1930}. It can be used to model currency exchange rates \citep{Ball1994} and commodity prices \citep{Schwartz1997}. A major application of the Ornstein--Uhlenbeck process is the modeling of interest rates by the so-called Vasicek model \citep{Vasicek1977}. The Ornstein--Uhlenbeck process can also be utilized to model stochastic volatility of financial assets \citep{Barndorff-Nielsen2001}. Another application is the trading strategy called the pairs trading \citep{Elliott2005}.

The Ornstein--Uhlenbeck process can be utilized when analyzing financial high-frequency data. \cite{Engle2000} coined a term ultra-high-frequency data referring to irregularly spaced financial data recorded for each transaction. Although this data can be aggregated to some fixed frequency (e.g.\  one minute), it is not without a loss of information. \cite{Ait-Sahalia2005} suggest to use as many observations as possible. In general, high-frequency time series, whether irregularly or regularly spaced, exhibit specific characteristics such as heavy tails, jumps, long-term memory, and market microstructure noise. In the Ornstein--Uhlenbeck model, heavy tails and jumps are often captured by generalizing the background driving process to the L\'{e}vy process \citep{Barndorff-Nielsen2001}. Alternatively, the background driving process can be generalized to the fractional Brownian motion to capture long-term memory \citep{Cheridito2003}.

We focus on challanges surrounding the market microstructure noise. Generally, logarithmic prices are assumed to follow a semimartingale \citep{Delbaen1994}. However, when the prices are observed at higher frequencies, it is evident that the semimartingale is contaminated by the market microstructure noise. This noise has relatively small variance but makes standard measures of volatility such as the realized variance significantly biased. Causes of the market microstructure noise include the bid-ask bounce, discretness of price values, discretness of price changes, informational effects and recording errors. Generally, the noise has a rich structure such as dependency on the price process and dependency in time (see e.g.\ \citealp{Hansen2006}). Many methods estimating quadratic variation and integrated variance robust to the noise were proposed in the high-frequency literature. Among nonparametric methods belong the two-scale estimator of \cite{Zhang2005}, the realized kernel estimator of \cite{Barndorff-Nielsen2008}, and the pre-averaging estimator of \cite{Jacod2009}. See e.g.\ \cite{Holy2023a} for their comparison. While a nonparametric approach is dominant in the literature, a parametric method was used for estimation of the Wiener process parameters by \cite{Ait-Sahalia2005}.

In this paper, we estimate the parameters of the Gaussian Ornstein--Uhlenbeck process in the presence of the independent Gaussian noise. We show that the Ornstein--Uhlenbeck parameters estimated by methods ignoring the noise are biased and inconsistent. As we argue, this is caused by the fact that the Ornstein--Uhlenbeck process contaminated by the independent Gaussian white noise and observed at discrete equidistant times follows ARMA(1,1) process instead of AR(1) process. We make use of this finding and propose a noise-robust estimator based on the ARMA(1,1) reparametrization. We also deal with the situation when the observations are not equidistant and propose a noise-robust estimator based on the maximum likelihood. For intial estimates, we utilize the method of moments with a noise-robust specification.

As an application of the Ornstein--Uhlenbeck process, we analyze the pairs trading strategy. This allows us to evaluate the added value of the noise-robust estimators compared to the traditional noise-sensitive estimators in terms of profit. The idea behind pairs trading lies in taking an advantage of financial markets that are out of equilibrium. When some pairs of prices exhibit strong similarity in the long run and they are currently far enough from their equilibrium, traders might profit by taking a long position in one security and a short position in the other security in a predetermined ratio. When the price spread reverts back to its mean level, the positions are closed and the profit is made. Typically, two similar commodities (e.g.\  West Texas Intermediate crude oil and Brent crude oil) or two stocks of companies in the same industry (e.g.\  Coca-Cola company and Pepsi company) are traded. The pairs trading can be further generalized to trading of groups of securities. There are three commonly used approaches in pairs trading: the distance approach \citep{Gatev2006a, Bowen2010, Rinne2017}, the cointegration approach \citep{Vidyamurthy2004a, Peters2011, Miao2014}, and the stochastic spread approach \citep{Elliott2005, Cummins2012, Goncu2016}. The focus of the stochastic spread approach is more on the time series analysis of a given pair of securities rather than the selection of securities. Typically, the spread process is modeled by a mean-reverting autoregressive process in discrete time or the Ornstein--Uhlenbeck process in continuous time. Entry and exit signals are then generated in the optimal way. For a comprehensive review of the pairs trading literature, see \cite{Krauss2017}. 

Some studies focus on intraday pairs trading. Namely, \cite{Bowen2010} analyze 60-minute data, \cite{Dunis2000} 30-minute data, \cite{Miao2014} 15-minute data, \cite{Peters2011} 10-minute data, and \cite{Liu2017} 5-minute data. However, none of these studies utilizes ultra-high-frequency data. Our aim is therefore to bring an insight into the pairs trading strategy in the context of ultra-high-frequency data.

We follow the stochastic spread approach based on the Ornstein--Uhlenbeck process. As in \cite{Bertram2009, Bertram2010}, we find the optimal trading signals using first-passage times of the process. While \cite{Bertram2009, Bertram2010} optimizes the strategy in terms of the maximum expected return and the maximum Sharpe ratio, we propose a strategy based on the mean-variance optimization. Note that \cite{Holy2022c} follow our work and further investigate properties of the proposed mean-variance optimization. In our study, we analyze stocks of 7 Big Oil companies traded on New York Stock Exchange (NYSE). We demonstrate that even when the variance of the noise is relatively small and one would simply decide to ignore it, which is unfortunatelly quite common in practice, it has a great impact on the estimated parameters. The reliance of market participants on this biased estimates can lead to wrong decisions and have harmful consequences. The pitfall of this lies in the fact that estimated parameters might appear as reliable values at the first sight but they are actually multiple times higher than their true values. We find that the use of the proposed estimator of the Ornstein--Uhlenbeck process with the correct treatment of the market microstructure noise leads to a significant increase in profitability of the pairs trading strategy.

The paper is structured as follows. In Section \ref{sec:ou}, we outline basic properties of the Ornstein--Uhlenbeck process, propose three noise-robust estimators, and compare them in a simulation study. In Section \ref{sec:pairs}, we compute the first passage times of the Ornstein--Uhlenbeck process and present the pairs trading strategy based on the mean-variance optimization. In Section \ref{sec:oil}, we illustrate the bias of the traditional noise-sensitive estimators and benefits of the proposed noise-robust estimators in an empirical study of 7 Big Oil companies. We conclude the paper in Section \ref{sec:conclusion}.

\section{Estimators of Ornstein--Uhlenbeck Process}
\label{sec:ou}

The Ornstein--Uhlenbeck process $P_t$, $ t \geq 0$ is a process satisfying stochastic differential equation
\begin{equation}
\label{eq:ouDiff}
\mathrm{d} P_t = \tau (\mu - P_t) \mathrm{d} t + \sigma \mathrm{d} W_t,
\end{equation}
where $W_t$ is a Wiener process, $\mu$ is a parameter representing long-term mean, $\tau > 0$ is a parameter representing speed of reversion, and $\sigma > 0$ is a parameter representing instantaneous volatility. This stochastic differential equation has solution
\begin{equation}
\label{eq:ouSolution}
P_t = P_0 e^{-\tau t} + \mu (1 - e^{-\tau t}) + \sigma \int_{0}^{t} e^{-\tau (t - s)} \mathrm{d} W_s.
\end{equation}
When assuming $P_0 \sim \mathrm{N} ( \mu, \sigma^2 / 2 \tau )$ and $P_0 \perp W_t$, $t \geq 0$, the Ornstein--Uhlenbeck process $P_t$ is a stationary process with normally distributed increments and unconditional moments
\begin{equation}
\label{eq:ouMomentsUnc}
\begin{aligned}
\mathrm{E} [ P_t ] &= \mu, \\
\mathrm{var} [ P_t ] &= \frac{\sigma^2}{2 \tau}, \\
\mathrm{cov} [ P_t, P_s ] &= \frac{\sigma^2}{2 \tau} e^{-\tau \lvert t-s \rvert}, \quad t \neq s. \\
\end{aligned}
\end{equation}
For a given initial value $p_0$, the Ornstein--Uhlenbeck process $P_t$ is a nonstationary process with normally distributed increments and conditional moments
\begin{equation}
\label{eq:ouMomentsCon}
\begin{aligned}
\mathrm{E} [ P_t  \mid P_0 = p_0 ] &= p_0 e^{-\tau t} + \mu \left( 1 - e^{-\tau t} \right), \\
\mathrm{var} [ P_t  \mid P_0 = p_0 ] &= \frac{\sigma^2}{2 \tau} \left( 1 - e^{-2 \tau t} \right), \\
\mathrm{cov} [ P_t, P_s \mid P_0 = p_0 ] &= \frac{\sigma^2}{2 \tau} \left( e^{-\tau \lvert t-s \rvert} - e^{-\tau (t+s)} \right), \quad t \neq s. \\
\end{aligned}
\end{equation}

In practice, we do not observe continuous paths of the process. Instead, we only observe the process $P_{T_i}$ at a finite number of discrete times $0 = T_0 < T_1 < \ldots < T_n = 1$, where $T_i$ are times of observations, which we treat as deterministic. Without loss of generality, we restrict ourselves to the time interval $[0,1]$. We further assume that the observed process is contaminated by independent white noise $E_i \sim \textrm{N}(0, \omega^2)$. For the observed discrete process $X_i$, we utilize the additive noise model
\begin{equation}
\label{eq:ouObsProcess}
X_{i} = P_{T_i} + E_{i}, \qquad i=0,\ldots,n.
\end{equation}
When assuming $P_0 \sim \mathrm{N} ( \mu, \sigma^2 / 2 \tau )$ and $P_0$ independent of $W_{T_i}$, $i \geq 0$, the observed process $X_i$ is a stationary process with normally distributed increments and unconditional moments
\begin{equation}
\label{eq:ouObsMomentsUnc}
\begin{aligned}
\mathrm{E} [ X_i ] &= \mu, \\
\mathrm{var} [ X_i ] &= \frac{\sigma^2}{2 \tau} + \omega^2, \\
\mathrm{cov} [ X_i, X_j ] &= \frac{\sigma^2}{2 \tau} e^{-\tau \lvert T_i-T_j \rvert}, \quad i \neq j. \\
\end{aligned}
\end{equation}
For a given $x_0$ the observed process $X_i$ is a nonstationary process with normally distributed increments and conditional moments
\begin{equation}
\label{eq:ouObsMomentsCon}
\begin{aligned}
\mathrm{E} [ X_i \mid X_0 = x_0] &= \mathrm{E}[ P_0\mid X_0 = x_0 ] e^{-\tau T_i} + \mu \left( 1 - e^{-\tau T_i} \right), \\
\mathrm{var} [ X_i \mid X_0 = x_0 ] &= \mathrm{var}[ P_0\mid X_0 = x_0 ] e^{-2 \tau T_i} + \frac{\sigma^2}{2 \tau} \left( 1 - e^{-2 \tau T_i} \right) + \omega^2, \\
\mathrm{cov} [ X_i, X_j \mid X_0 = x_0 ] &= \mathrm{var}[ P_0\mid X_0 = x_0] e^{-\tau (T_i+T_j)} \\
& \quad + \frac{\sigma^2}{2 \tau} \left( e^{-\tau \lvert T_i-T_j \rvert} - e^{-\tau (T_i+T_j)} \right), \quad i \neq j, \\
\end{aligned}
\end{equation}
where
\begin{equation}
\label{eq:cinit}
\begin{aligned}
\mathrm{E} [ P_0 \mid X_0 = x_0 ] &= \frac{x_0 \sigma^2 + 2 \tau \mu \omega^2}{\sigma^2 + 2 \tau \omega^2}, \\
\mathrm{var} [ P_0 \mid X_0 = x_0 ] &= \frac{\sigma^2 \omega^2}{\sigma^2 + 2 \tau \omega^2}. \\ 
\end{aligned}
\end{equation}
This conditional distribution is derived in Appendix \ref{app:conditional}.

Let us analyze the situation in which we assume observations to follow the Ornstein--Uhlenbeck process $P_{T_i}$ but they actually follow the noisy process $X_i$. From \eqref{eq:ouMomentsUnc} and \eqref{eq:ouObsMomentsUnc} we have unconditional moments
\begin{equation}
\label{eq:ouMomentsUncNR}
\begin{aligned}
\mathrm{E} [ X_i ] &= \mathrm{E} [ P_{T_i} ], \\
\mathrm{var} [ X_i ] &= \mathrm{var} [ P_{T_i} ] + \omega^2, \\
\mathrm{cov} [ X_i, X_j ] &= \mathrm{cov} [ P_{T_i}, P_{T_j} ], \quad i \neq j. \\
\end{aligned}
\end{equation}
This means that an unbiased estimate of the expected value of $X_i$ is also an unbiased estimate of the expected value of $P_{T_i}$. The same applies for the autocovariance function of $X_i$ and the autocovariance function of $P_{T_i}$. An unbiased estimate of the variance of $X_i$, on the contrary, is a positively biased estimator of the variance of $P_{T_i}$. Because of this, the autocorrelation function
\begin{equation}
\mathrm{cor} [ X_i, X_j ] = \mathrm{cor} [ P_{T_i}, P_{T_j} ] - \frac{2 \tau \omega^2}{\sigma^2 + 2 \tau \omega^2} e^{- \tau \lvert T_i-T_j \rvert}, \quad i \neq j
\end{equation}
also differes from the autocorrelation function of $P_{T_i}$. To sum up, the misspecification of the process does not affect unconditional expected value and autocovariance estimation, but does affect unconditional variance and autocorrelation estimation. 

Our goal is to estimate the parameters $\mu$, $\tau$, $\sigma$ of the Ornstein--Uhlenbeck process $P_{T_i}$ and the parameter $\omega$ of the market microstructure noise $E_i$ from the observed process $X_{i}$. For this purpose, we propose the method of moments estimator, the maximum likelihood estimator, and the estimator reparametrizing discretized Ornstein--Uhlenbeck process with the noise as an ARMA(1,1) process.

\subsection{Method of Moments}
\label{sec:ouMoments}

The method of moments is based on relating theoretical values of random variable moments to their finite-sample estimates. The advantage of the method of moments lies in its simplicity and closed-form solution. It is often used as an initial solution for more sofisticated methods such as the maximum likelihood estimator.

We propose the method of moments estimator considering the market microstructure noise. In this section, we assume that the times of observations $T_i$ are equally spaced and $T_i - T_{i-1} = n^{-1}$. As we estimate four parameters of the observed process $X_i$, we utilize four unconditional moments $\textrm{E} [ X_i ]$, $\textrm{var} [ X_i ]$, $\textrm{cov} [ X_i, X_{i-1} ]$ and $\textrm{cov} [ X_i, X_{i-2} ]$ specified in \eqref{eq:ouMomentsUncNR}. We can estimate these moments using observed values $x_0, x_1, \ldots, x_n$ as
\begin{equation}
\begin{aligned}
M_{1,n} &= \frac{1}{n+1} \sum_{i=0}^n X_i, \\
M_{2,n} &= \frac{1}{n} \sum_{i=0}^n (x_i - M_{1,n})^2, \\
M_{3,n} &= \frac{1}{n-1} \sum_{i=1}^n (x_i - M_{1,n})(x_{i-1} - M_{1,n}), \\
M_{4,n} &= \frac{1}{n-2} \sum_{i=2}^n (x_i - M_{1,n})(x_{i-2} - M_{1,n}). \\
\end{aligned}
\end{equation}
By solving equations
\begin{equation}
\textrm{E} [ X_i ] = M_{1,n}, \quad \textrm{var} [ X_i ] = M_{2,n}, \quad \textrm{cov} [ X_i, X_{i-1} ] = M_{3,n}, \quad \textrm{cov} [ X_i, X_{i-2} ] = M_{4,n},
\end{equation}
we get estimates
\begin{equation}
\begin{aligned}
\hat{\mu} &= M_{1,n}, \\
\hat{\tau} &= \frac{1}{\Delta} \log \frac{M_{3,n}}{M_{4,n}}, \\
\hat{\sigma}^2 &= 2 \frac{1}{\Delta} \frac{M_{3,n}^2}{M_{4,n}} \log \frac{M_{3,n}}{M_{4,n}}, \\
\hat{\omega}^2 &= M_{2,n} - \frac{M_{3,n}^2}{M_{4,n}}. \\
\end{aligned}
\end{equation}
Higher moments and higher lags of autocovariance function can also be used. However, because we use this method mainly as initial estimates, we do not focus on finding the optimal set of moments. The method of moments estimator not assuming the noise and its bias in the presence of the noise is derived in Appendix \ref{app:moments}.

\subsection{Maximum Likelihood Method}
\label{sec:ouLikelihood}

A widely used method for parameter estimation is the maximum likelihood estimator. It maximizes the likelihood function (or, equivalently, the logarithmic likelihood function) given the observations. In our case, it utilizes the normal conditional density function for the Ornstein--Uhlenbeck process. In some simple cases, the maximum likelihood estimators are available in a closed form. \cite{Tang2009} present the closed-form estimates for the regularly spaced Ornstein--Uhlenbeck process without the noise.

We focus on the more general case of the irregularly spaced Ornstein--Uhlenbeck process contaminated by the noise. The maximum likelihood estimates are obtained by maximizing the logarithmic likelihood function given by
\begin{equation}
L(\mu, \tau, \sigma^2, \omega^2) = \sum_{i=1}^n \log f_{X_{i}} \left( x_{i} \mid X_{i-1} = x_{i-1} \right),
\end{equation}
where $f_{X_i} \left( x_{i} \mid X_{i-1} = x_{i-1} \right)$ is the conditional density function of the observations. According to \eqref{eq:ouObsMomentsCon} and Appendix \ref{app:conditional}, it is the conditional density function of the normal distribution
\begin{equation}
\begin{aligned}
f_{X_{i}}(x_{i} \mid X_{i-1} = x_{i-1}) &= \frac{1}{\sqrt{2 \pi \mathrm{var}[X_{i} \mid X_{i-1} = x_{i-1}]}} \\
& \quad \times \exp \left\{-\frac{\left( x_{i} - \mathrm{E}[X_{i} \mid X_{i-1} = x_{i-1}] \right)^2}{2\mathrm{var}[X_{i} \mid X_{i-1} = x_{i-1}]} \right\}
\end{aligned}
\end{equation}
with conditional moments
\begin{equation}
\begin{aligned}
\mathrm{E} [ X_{i} \mid X_{i-1} = x_{i-1} ] &= \frac{x_{i-1} \sigma^2 + 2 \tau \mu \omega^2}{\sigma^2 + 2 \tau \omega^2} e^{-\tau \left(T_i - T_{i-1} \right)} + \mu \left( 1 - e^{-\tau \left(T_i - T_{i-1} \right)} \right), \\
\mathrm{var} [ X_{i} \mid X_{i-1} = x_{i-1} ] &= \frac{\sigma^2 \omega^2}{\sigma^2 + 2 \tau \omega^2} e^{-2 \tau \left(T_i - T_{i-1} \right)} + \frac{\sigma^2}{2 \tau} \left( 1 - e^{-2 \tau \left(T_i - T_{i-1} \right)} \right) + \omega^2. \\
\end{aligned}
\end{equation}
The logarithmic likelihood function can be simplified to
\begin{equation}
\begin{aligned}
L(\mu, \tau, \sigma^2, \omega^2) &= -\frac{1}{2} \sum_{i=1}^n \log \left( 2 \pi \mathrm{var}[X_{i} \mid X_{i-1} = x_{i-1} ] \right) \\
& \quad - \frac{1}{2} \sum_{i=1}^n \frac{\left( x_t - \mathrm{E}[X_{i} \mid X_{i-1} = x_{i-1} ] \right)^2}{\mathrm{var}[X_{i} \mid X_{i-1} = x_{i-1} ]}.
\end{aligned}
\end{equation}
The estimates are then given by
\begin{equation}
(\hat{\mu}, \hat{\tau}, \hat{\sigma}^2, \hat{\omega}^2)' = \underset{\mu, \tau, \sigma^2, \omega^2}{\mathrm{arg} \max} \  L(\mu, \tau, \sigma^2, \omega^2) \quad \text{s. t.} \quad \sigma^2 \geq 0, \ \omega^2 \geq 0.
\end{equation}
In the sequel, we find the optimal solution iteratively by the Sbplx algorithm of \cite{Rowan1990}, a variant of the Nelder-Mead algorithm, and use the method of moments for an intial solution. The estimator for the irregularly spaced Ornstein--Uhlenbeck process without the noise is derived in Appendix \ref{app:likelihood}.

\subsection{ARMA Reparametrization}
\label{sec:ouArma}

The ARMA reparametrization lies in the following three steps. First, we reparametrize the discretized equidistant process to a commonly used and studied time series model. Second, we estimate parameters of the time series model, e.g.\  by the conditional-sum-of-squares or maximum likelihood estimators. Third, we transform the estimates back to the original parametrization. A possible disadvantage is that the reparametrization does not respect parameter restrictions. In our case, $\sigma^2$ and $\omega^2$ parameters should be non-negative, but the reparametrization allows for negative values.

It is well known that the discretized Ornstein--Uhlenbeck process corresponds to an AR(1) process. \cite{Ait-Sahalia2005} reparametrized the discretized Wiener process contaminated by the white noise as an ARIMA(0,1,1) process. As the discretized Wiener process without the noise is an ARIMA(0,1,0) process, the noise therefore induces a moving average component of order one. We show that the same applies for the discretized Ornstein--Uhlenbeck process contaminated by the white noise as it corresponds to an ARMA(1,1) process.

In this section, we assume the times of observations $T_i$ are equally spaced and denote $\Delta = T_i - T_{i-1} = n^{-1}$. Using \eqref{eq:ouSolution} with initial time $T_{i-1}$, the process $X_i$ can be decomposed as
\begin{equation}
\label{eq:ouArmaFirst}
\begin{aligned}
X_i &= P_{T_i} + E_{i} \\
&= \mu (1 - e^{-\tau \Delta}) + P_{T_{i-1}} e^{-\tau \Delta} + \sigma \int_{0}^{\Delta} e^{-\tau (\Delta - s)} \mathrm{d} W_s + E_i \\
&= \mu (1 - e^{-\tau \Delta}) + X_{i-1} e^{-\tau \Delta} + \sigma \int_{0}^{\Delta} e^{-\tau (\Delta - s)} \mathrm{d} W_s + E_i - E_{i-1} e^{-\tau \Delta},
\end{aligned}
\end{equation}
where the last equality holds because $P_{T_{i-1}} = X_{i-1} - E_{i-1}$. We denote
\begin{equation}
\label{eq:ouArmaReparDet}
\begin{aligned}
\alpha &= \mu (1 - e^{-\tau \Delta}), \\
\varphi &= e^{-\tau \Delta}. \\
\end{aligned}
\end{equation}
We further denote
\begin{equation}
\label{eq:ouArmaReparStoch}
U_i = \sigma \int_{0}^{\Delta} e^{-\tau (\Delta - s)} \mathrm{d} W_s + E_i - E_{i-1} e^{-\tau \Delta}.
\end{equation}
Using \eqref{eq:ouObsMomentsCon} we have that the random variable $U_i$ is normally distributed with moments
\begin{equation}
\label{eq:ouArmaReparMom}
\begin{aligned}
\textrm{E}[U_i] &= 0, \\
\textrm{var}[U_i] &= \frac{\sigma^2}{2 \tau} (1 - e^{-2 \tau \Delta})  + \omega^2 (1 + e^{-2 \tau \Delta}), \\
\textrm{cov}[U_i, U_{i-1}] &= - \omega^2 e^{-\tau \Delta}, \\
\textrm{cov}[U_i, U_{i-j}] &= 0, \qquad j > 1. \\
\end{aligned}
\end{equation}
Using substitutions \eqref{eq:ouArmaReparDet} and \eqref{eq:ouArmaReparStoch}, we rewrite \eqref{eq:ouArmaFirst} as
\begin{equation}
\label{eq:ouArmaSecond}
X_i = \alpha + \varphi X_{i-1} + U_i.
\end{equation}
Let us define a moving average process of order one $\tilde{U}_i$, $i \geq 0$ as 
\begin{equation}
\label{eq:ouArmaReparMa}
\tilde{U}_i = \theta V_{i-1} + V_i, \qquad V_i \stackrel{i.i.d.}{\sim} \textrm{N}(0, \gamma^2).
\end{equation}
Variable $\tilde{U}_i$ is then normally distributed with moments
\begin{equation}
\label{eq:ouArmaReparMaMom}
\begin{aligned}
\textrm{E}[\tilde{U}_i] &= 0, \\
\textrm{var}[\tilde{U}_i] &= \gamma^2 (1 + \theta^2), \\
\textrm{cov}[\tilde{U}_i, \tilde{U}_{i-1}] &= \theta \gamma^2. \\
\textrm{cov}[\tilde{U}_i, \tilde{U}_{i-j}] &= 0, \qquad j > 1. \\
\end{aligned}
\end{equation}
We show that the process $\{ U_i \}_{i \geq 0}$ is equivalent to the process $\{ \tilde{U}_i \}_{i \geq 0}$ for the right choice of $\gamma$ and $\theta$ parameters satisfying
\begin{equation}
\begin{aligned}
\label{eq:ouArmaUu}
\textrm{var}[U_i] &= \textrm{var}[\tilde{U}_i], \\
\textrm{cov}[U_i, U_{i-1}] &= \textrm{cov}[\tilde{U}_i, \tilde{U}_{i-1}]. \\
\end{aligned}
\end{equation}
The joint distribution of the process $\{ U_i \}_{i \geq 0}$ is identical to the joint distribution of the process $\{ \tilde{U}_i \}_{i \geq 0}$ as both processes are normally distributed with zero first moment and the same autocovariation function. We can then rewrite \eqref{eq:ouArmaSecond} as 
\begin{equation}
\label{eq:ouArmaThird}
X_i = \alpha + \varphi X_{i-1} + \tilde{U}_i.
\end{equation}
This is an ARMA(1,1) process of the form
\begin{equation}
\label{eq:ouArmaFourth}
\begin{aligned}
X_i = \alpha + \varphi X_{i-1} + \theta V_{i-1} + V_i, \qquad V_i \stackrel{i.i.d.}{\sim} \textrm{N}(0, \gamma^2).
\end{aligned}
\end{equation}
We can estimate parameters $\alpha$, $\varphi$, $\theta$ and $\gamma^2$ by any suitable method. Substitution \eqref{eq:ouArmaReparDet} and equivalency \eqref{eq:ouArmaUu} with \eqref{eq:ouArmaReparMom} and \eqref{eq:ouArmaReparMaMom} imply
\begin{equation}
\label{eq:ouArmaSystem}
\begin{aligned}
\hat{\alpha} &= \hat{\mu} (1 - e^{-\hat{\tau} \Delta}), \\
\hat{\varphi} &=  e^{-\hat{\tau} \Delta}, \\
\hat{\gamma}^2 (1 + \hat{\theta}^2) &= \frac{\hat{\sigma}^2}{2 \hat{\tau}} (1 - e^{-2 \hat{\tau} \Delta})  + \hat{\omega}^2 (1 + e^{-2 \hat{\tau} \Delta}), \\
\hat{\theta} \hat{\gamma}^2 &= - \hat{\omega}^2 e^{-\hat{\tau} \Delta}. \\
\end{aligned}
\end{equation}
Finally, by solving this system of equations, we get estimates
\begin{equation}
\begin{aligned}
\hat{\mu} &= \frac{\hat{\alpha}}{1 - \hat{\varphi}}, \\
\hat{\tau} &= - \frac{1}{\Delta} \log \hat{\varphi}, \\
\hat{\sigma}^2 &= -2 \frac{1}{\Delta} \frac{\hat{\gamma}^2 (\hat{\varphi} + \hat{\theta}^2 \hat{\varphi} + \hat{\theta} \hat{\varphi}^2 + \hat{\theta})}{\hat{\varphi} (1 -\hat{\varphi}^2)} \log \hat{\varphi}, \\
\hat{\omega}^2 &= - \frac{\hat{\theta} \hat{\gamma}^2}{\hat{\varphi}}. \\
\end{aligned}
\end{equation}
The case without the noise is presented in Appendix \ref{app:ar}.

\subsection{Simulation Study}
\label{sec:ouSimulation}

We evaluate the finite-sample performance of the proposed estimators using simulations. We simulate the observed price process as the Ornstein--Uhlenbeck process with parameters $\mu = 10^{-1}$, $\tau = 10$ and $\sigma^2 = 10^{-4}$ contaminated by the independent Gaussian white noise with variance $\omega^2=10^{-8}$. We select the values of parameteres to resemble values reported in the empirical study in Section \ref{sec:oil}. The simulated observations are irregularly spaced and the times of observations are generated by the Poisson point process. We perform the simulation \num{1 000 000} times, each with \num{23 400} observations. The number of observations corresponds to durations between price changes to be one second on average during 6.5 hours long trading day. For the simulation of the Ornstein--Uhlenbeck process, we adopt the exact simulation algorithm (see e.g.\ \citealp{Glasserrnan2004}).

We compare the estimators by mean absolute errors of estimated parameters. The noise-sensitive method of moments is denoted as 1MIN-MOM and its noise-robust modification as 1MIN-MOM-NR. The approach based on the reparametrization to time series models estimates parameters by the conditional sum-of-squares and is denoted as 1MIN-AR for the noise-sensitive reparametrization to the AR(1) process and 1-MIN-ARMA-NR for the noise-robust reparametrization to the ARMA(1,1) process. The noise-sensitive and noise-robust maximum likelihood estimators based on 1-minute data are denoted as 1MIN-MLE and 1MIN-MLE-NR respectively while their tick-data counterparts are denoted as TICK-MLE and TICK-MLE-NR respectively.

The variance of the process can also be estimated by nonparametric methods. Since the parameter $\sigma^2$ of the Ornstein--Uhlenbeck process is equal to the quadratic variation of the process over time interval $(0, 1)$, we can estimate $\sigma^2$ by nonparametric estimators of quadratic variation. The straightforward estimator of quadratic variation is the realized variance. However, as shown for example by \cite{Hansen2006}, it is biased and inconsistent in the presence of the market microstructure noise. We denote the realized variance based on 1-minute data as 1MIN-RV and TICK-RV for tick data. There are many noise-robust alternatives for the nonparametric quadratic variation estimation in the literature. One of the method is the realized kernel estimator proposed by \cite{Barndorff-Nielsen2008}. We utilize the variant with the modified Tukey-Hanning kernel and denote it as 1MIN-RK-TH2 for 1-minute data and TICK-RK-TH2 for tick data. Another noise-robust method is the pre-averaging estimator of \cite{Jacod2009}. It is denoted as 1MIN-PAE for 1-minute data and TICK-PAE for tick data. The variance of the noise $\omega^2$ is estimated using biased realized variance $RV_n$ adjusted for the noise-robust estimate $RM_n$ (either the realized kernel or the pre-averaging estimate) $\hat{\omega}^2 = (RV_n - RM_n) / 2 n$, where $n$ is the number of observations.

The results of simulations are reported in Table \ref{tab:ouSim}. Generally, the noise-robust estimators based on tick data outperform the noise-robust estimators based on 1-minute data while the noise-sensitive estimators based on tick data are outperformed by the noise-sensitive estimators based on 1-minute data. This is because the noise-robust estimators can utilize the additional information from tick data while the noise-sensitive estimators are more biased with more observations. We further investigate this property in Figure \ref{fig:ouSignature} in the empirical study. When considering only 1-minute data, the best parametric estimator is the 1MIN-ARMA-NR. However, for the volatility estimation based on 1-minute data, nonparametric estimators 1MIN-RK-TH2 and 1MIN-PAE are superior to parametric estimators. When considering both tick data and 1-minute aggregation, the best parametric estimator is the TICK-MLE-NR. The shortcoming of this estimator is slightly worse estimation of $\mu$, but it is compensated by the lowest mean absolute error of $\tau$ and $\sigma^2$ parameters. On the other hand, its noise-sensitive variant TICK-MLE performs very poorly due to the misspecification of the process (omitting the noise). Interestingly, the TICK-MLE-NR even outperforms the nonparametric TICK-RK-TH2 and TICK-PAE estimators in the estimation of the variance $\sigma^2$. In the rest of the study, we work solely with tick data and focus only on the TICK-MLE and TICK-MLE-NR estimators.

\begin{table}
\centering
\begin{tabular}{lrrrr}
\toprule
Method  &                       $\mu$ &                $\tau$ &               $\sigma$ &               $\omega$ \\
\midrule
1MIN-MOM     & $7.5797 \cdot 10^{-4}$ & $0.4709 \cdot 10^{2}$ & $0.9945 \cdot 10^{-2}$ &                      - \\
1MIN-MOM-NR  & $7.5797 \cdot 10^{-4}$ & $0.2032 \cdot 10^{2}$ & $0.4515 \cdot 10^{-2}$ & $3.6843 \cdot 10^{-5}$ \\ 
1MIN-AR      & $7.5797 \cdot 10^{-4}$ & $0.4683 \cdot 10^{2}$ & $0.9906 \cdot 10^{-2}$ &                      - \\
1MIN-ARMA-NR & $7.5797 \cdot 10^{-4}$ & $0.1358 \cdot 10^{2}$ & $0.2783 \cdot 10^{-2}$ & $2.8069 \cdot 10^{-5}$ \\
1MIN-MLE     & $7.6020 \cdot 10^{-4}$ & $0.4683 \cdot 10^{2}$ & $0.9906 \cdot 10^{-2}$ &                      - \\
TICK-MLE     & $7.8099 \cdot 10^{-4}$ & $9.0415 \cdot 10^{2}$ & $8.8591 \cdot 10^{-2}$ &                      - \\
1MIN-MLE-NR  & $7.6020 \cdot 10^{-4}$ & $0.2042 \cdot 10^{2}$ & $0.4543 \cdot 10^{-2}$ & $3.7319 \cdot 10^{-5}$ \\
TICK-MLE-NR  & $7.5910 \cdot 10^{-4}$ & $0.0543 \cdot 10^{2}$ & $0.0263 \cdot 10^{-2}$ & $0.0658 \cdot 10^{-5}$ \\
1MIN-RV      &                      - &                     - & $0.9893 \cdot 10^{-2}$ &                      - \\
TICK-RV      &                      - &                     - & $1.3831 \cdot 10^{-2}$ &                      - \\
1MIN-RK-TH2  &                      - &                     - & $0.1392 \cdot 10^{-2}$ & $0.3271 \cdot 10^{-5}$ \\
TICK-RK-TH2  &                      - &                     - & $0.0797 \cdot 10^{-2}$ & $0.1821 \cdot 10^{-5}$ \\
1MIN-PAE     &                      - &                     - & $0.0315 \cdot 10^{-2}$ & $0.0826 \cdot 10^{-5}$ \\
TICK-PAE     &                      - &                     - & $0.0322 \cdot 10^{-2}$ & $0.0836 \cdot 10^{-5}$ \\
\bottomrule
\end{tabular}
\caption{Mean absolute errors of parameters estimated by various methods from the simulated noisy Ornstein--Uhlenbeck process with true parameters $\mu=1$, $\tau = 10$, $\sigma^2 = 10^{-4}$ and $\omega^2 = 10^{-8}$.}
\label{tab:ouSim}
\end{table}

\section{Optimal Pairs Trading Strategy}
\label{sec:pairs}

For a given pair of stocks A and B, the pairs trading strategy is based on the logarithmic price spread process
\begin{equation}
P_t = \ln \left( \frac{A_t}{B_t} \right) = \ln A_t - \ln B_t,
\end{equation}
where $A_t$ is the price of stock A and $B_t$ is the price of stock B. We model the process $P_t$ as the Ornstein--Uhlenbeck process given by \eqref{eq:ouDiff} with a long-term mean $\mu$, speed of reversion $\tau$ and instantaneous volatility $\sigma > 0$. The strategy itself consists of the following steps. First, we wait until the logarithmic price spread $P_t$ reaches a given entry level $a$ at time $t_1$. Without loss of generality, we assume the entry level $a$ is greater than the long-term mean $\mu$, i.e.\  $a > \mu$. When the entry level is reached, we simultaneously enter short position in stock A and long posistion in stock B. We expect the price of A to go down and price of B to go up, i.e.\  the spread to revert to its long-term mean. When the logarithmic price spread $P_t$ reaches a given exit level $b < a$ at time $t_2$, we clear both positions and make profit. The profit from stock A in terms of continuous compound rate of return is $\ln A_{t_1} - \ln A_{t_2}$ while the profit from stock B is $\ln B_{t_2} - \ln B_{t_1}$. Adding a transaction cost $c$ for the whole pairs trade, we have the total profit
\begin{equation}
\begin{aligned}
r &= \ln A_{t_1} - \ln A_{t_2} + \ln B_{t_2} - \ln B_{t_1} - c \\
&= P_{t_1} - P_{t_2} - c \\
&= a - b - c. \\
\end{aligned}
\end{equation}
After the trade, we again wait for the spread $P_t$ to reach the entry level $a$ and repeat the whole trading cycle. The trading cycle is thus composed of two parts. In the first part, we hold short and long positions in stocks A and B respectively, while in the second part, we wait until the next trading signal. We denote the duration of the trading cycle as
\begin{equation}
\label{eq:ouDurationCycle}
\mathcal{T} = \mathcal{T}_{a \to b} + \mathcal{T}_{b \to a},
\end{equation}
where $\mathcal{T}_{a \to b}$ is the first passage time from $a$ to $b$ and $\mathcal{T}_{b \to a}$ is the first passage time from $b$ to $a$.

In this strategy, we short stock A and long stock B. The opposite strategy can be adopted as well. In that case, when reaching the entry level $a' < \mu$, we long A and short B. Then, when reaching the exit level $b' > a'$, we make profit $b' - a' - c$. Since the Ornstein--Uhlenbeck process is symmetric around $\mu$, the second strategy for stocks A and B is identical to the first strategy for stocks B and A. For simplicity, we focus only on the first case for stocks A and B with $a > \mu$.

Our goal is to determine the values of entry signal $a$ and exit signal $b$ for a given transaction cost $c$ and static process parameters $\mu$, $\tau$ and $\sigma$. To optimally select signals $a$ and $b$, we closely follow the framework of \cite{Bertram2009} and \cite{Bertram2010}, also adopted by \cite{Cummins2012}, \cite{Zeng2014}, and \cite{Goncu2016}. All these papers focus on maximazing the expected profit while \cite{Bertram2010} also deals with maximazing the Sharpe ratio. In our work, we adopt the mean-variance optimization related to the modern portfolio theory. We formulate the problem as the maximization of the expected profit for a given level of maximum variance. If the level of maximum variance is large enough, the problem simply reduces to the maximization of the expected profit.

Let $Z_t$ be the random profit of the strategy over time $t$. For a given entry signal $a$, exit signal $b$ and transaction cost $c$, it is equal to
\begin{equation}
Z_t = \left(a - b - c \right) N_t,
\end{equation}
where $N_t$ is the counting process representing the number of trades during time $t$. Because the profit per trade $a - b - c$ is always constant, the only randomness lies in the process $N_t$. Further, let us define the expected profit per unit time and variance of profit per unit time as
\begin{equation}
\begin{aligned}
Z_M &= \lim_{t \to \infty} \frac{\mathrm{E}[Z_t]}{t} = \lim_{t \to \infty} \frac{\left(a - b - c \right) \mathrm{E} N_t}{t}, \\
Z_V &= \lim_{t \to \infty} \frac{\mathrm{var}[Z_t]}{t} = \lim_{t \to \infty} \frac{\left(a - b - c \right)^2 \mathrm{var} N_t}{t}. \\
\end{aligned}
\end{equation}
As in \cite{Bertram2010}, using the results from the renewal theory for the expected value and variance (see e.g.\  \citealp{Cox1965}), we obtain
\begin{equation}
\label{eq:ouStrategyRenewal}
\begin{aligned}
Z_M &= \frac{a - b - c}{\mathrm{E} \mathcal{T}}, \\
Z_V &= \frac{\left(a - b - c \right)^2 \mathrm{var} \mathcal{T} }{\left( \mathrm{E} \mathcal{T} \right)^3}, \\
\end{aligned}
\end{equation}
where $\mathcal{T}$ is the trading cycle duration given by \eqref{eq:ouDurationCycle}. In our mean-variance optimization, we utilize these two moments per unit time.

\subsection{Dimensionless System}
\label{sec:pairsDimension}

Following \cite{Bertram2010} and \cite{Zeng2014}, we reparametrize the Ornstein--Uhlenbeck process \eqref{eq:ouDiff} to the dimensionless system. We transform the process to
\begin{equation}
\label{eq:ouRepar}
\tilde{P}_t = \sqrt{\frac{2 \tau}{\sigma^2}} \left( P_t - \mu \right),
\end{equation}
and perform the time dilation $\tilde{t} = \tau t$. Using It\^{o}'s lemma, we have
\begin{equation}
\mathrm{d} \tilde{P}_{\tilde{t}} = - \tilde{P}_{\tilde{t}} \mathrm{d} \tilde{t} + \sqrt{2} \mathrm{d} W_{\tilde{t}}.
\end{equation}
A major advantage of this reparametrization is that it does not depend on parameters $\mu$, $\tau$ and $\sigma^2$. For this reason, the subsequent analysis of first passage times and optimal signals is much more simple. The dimensionless system also allows us to study the impact of biased parameters on the pairs trading strategy. The reparametrized entry level, exit level and transaction cost are respectively
\begin{equation}
\label{eq:ouReparLev}
\begin{aligned}
\tilde{a} &= \sqrt{\frac{2 \tau}{\sigma^2}} \left( a - \mu \right), \quad & a &= \sqrt{\frac{\sigma^2}{2 \tau}} \tilde{a} + \mu, \\
\tilde{b} &= \sqrt{\frac{2 \tau}{\sigma^2}} \left( b - \mu \right), \quad & b &= \sqrt{\frac{\sigma^2}{2 \tau}} \tilde{b} + \mu, \\
\tilde{c} &= \sqrt{\frac{2 \tau}{\sigma^2}} c, \quad & c &= \sqrt{\frac{\sigma^2}{2 \tau}} \tilde{c}. \\
\end{aligned}
\end{equation}
The reparametrized duration of trading cycle is
\begin{equation}
\label{eq:ouReparDur}
\begin{aligned}
\tilde{\mathcal{T}} &= \tau \mathcal{T}, \quad & \mathcal{T} &= \frac{1}{\tau} \tilde{\mathcal{T}}. \\
\end{aligned}
\end{equation}
Finally, the reparametrized expected profit per unit time and variance of profit per unit time are respectively
\begin{equation}
\label{eq:ouReparMom}
\begin{aligned}
\tilde{Z}_M &= \sqrt{\frac{2}{\tau \sigma^2}} Z_M, \quad & Z_M &= \sqrt{\frac{\tau \sigma^2}{2}} \tilde{Z}_M, \\
\tilde{Z}_V &= \frac{2}{\sigma^2} Z_V, \quad & Z_V &= \frac{\sigma^2}{2} \tilde{Z}_V. \\
\end{aligned}
\end{equation}

\subsection{First Passage Times}
\label{sec:pairsPassage}

The key variable in expression for moments per time \eqref{eq:ouStrategyRenewal} is the duration of trading cycle. In the dimensionless system, it is equal to
\begin{equation}
\tilde{\mathcal{T}} = \tilde{\mathcal{T}}_{\tilde{a} \to \tilde{b}} + \tilde{\mathcal{T}}_{\tilde{b} \to \tilde{a}}.
\end{equation}
When assuming $\tilde{a} > 0$ and $\tilde{b} < \tilde{a}$, it is the sum of the first passage time from $\tilde{a}$ to $\tilde{b}$ and the first passage time from $\tilde{b}$ to $\tilde{a}$ defined as
\begin{equation}
\begin{aligned}
\tilde{\mathcal{T}}_{\tilde{a} \to \tilde{b}} &= \inf \left\{ t: \tilde{P}_t < \tilde{b} \mid \tilde{P}_0 = \tilde{a} \right\}, \\
\tilde{\mathcal{T}}_{\tilde{b} \to \tilde{a}} &= \inf \left\{ t: \tilde{P}_t > \tilde{a} \mid \tilde{P}_0 = \tilde{b} \right\}. \\
\end{aligned}
\end{equation}
In this section, we present the expected value and variance of the trading cycle duration. These results are based on the explicit expressions of the first-passage-time moments derived by \cite{Ricciardi1988}. We denote the gamma function as $\Gamma(\cdot)$ and digamma function as $\psi(\cdot)$.

The expected values of the first passage times from $\tilde{a}$ to $\tilde{b}$ and from $\tilde{b}$ to $\tilde{a}$ are respectively
\begin{equation}
\begin{aligned}
\mathrm{E} \tilde{\mathcal{T}}_{\tilde{a} \to \tilde{b}} &= \phi_1(-\tilde{b}) - \phi_1(-\tilde{a}), \\
\mathrm{E} \tilde{\mathcal{T}}_{\tilde{b} \to \tilde{a}} &= \phi_1(\tilde{a}) - \phi_1(\tilde{b}), \\
\end{aligned}
\end{equation}
where
\begin{equation}
\label{eq:ouPassagePhi1}
\phi_1(z) = \frac{1}{2} \sum_{k=1}^{\infty} \frac{ \left( \sqrt{2} z \right)^k }{ k! } \Gamma \left( \frac{k}{2} \right).
\end{equation}
The expected value of the trading cycle duration is then
\begin{equation}
\label{eq:ouCycleMean}
\mathrm{E} \tilde{\mathcal{T}} = \sum_{k=1}^{\infty} \frac{ \left( \sqrt{2} \tilde{a} \right)^{2k - 1} - \left( \sqrt{2} \tilde{b} \right)^{2k - 1} }{ \left(2k - 1 \right) ! }  \Gamma \left( \frac{2k - 1}{2} \right).
\end{equation}
The variances of the first passage times from $\tilde{a}$ to $\tilde{b}$ and from $\tilde{b}$ to $\tilde{a}$ are respectively
\begin{equation}
\begin{aligned}
\mathrm{var} \tilde{\mathcal{T}}_{\tilde{a} \to \tilde{b}} &= \left( \phi_1(-\tilde{b}) \right)^2 - \phi_2(-\tilde{b}) + \phi_2(-\tilde{a}) - \left( \phi_1(-\tilde{a}) \right)^2, \\
\mathrm{var} \tilde{\mathcal{T}}_{\tilde{b} \to \tilde{a}} &= \left( \phi_1(\tilde{a}) \right)^2 - \phi_2(\tilde{a}) + \phi_2(\tilde{b}) - \left( \phi_1(\tilde{b}) \right)^2, \\
\end{aligned}
\end{equation}
where $\phi_1(z)$ is given by \eqref{eq:ouPassagePhi1} and
\begin{equation}
\label{eq:ouPassagePhi2}
\phi_2(z) = \frac{1}{2} \sum_{k=1}^{\infty} \frac{ \left( \sqrt{2} z \right)^k }{ k! } \Gamma \left( \frac{k}{2} \right) \left( \psi \left( \frac{k}{2} \right) - \psi \left( 1 \right) \right).
\end{equation}
The variance of the trading cycle duration is then
\begin{equation}
\label{eq:ouCycleVar}
\mathrm{var} \tilde{\mathcal{T}} = w_1 (\tilde{a}) - w_1 (\tilde{b}) - w_2 (\tilde{a}) + w_2 (\tilde{b}),
\end{equation}
where
\begin{equation}
\begin{aligned}
w_1(z) &= \left( \frac{1}{2} \sum_{k=1}^{\infty} \frac{ \left( \sqrt{2} z \right)^k }{ k! } \Gamma \left( \frac{k}{2} \right) \right)^2 - \left(  \frac{1}{2} \sum_{k=1}^{\infty} \frac{ \left( - \sqrt{2} z \right)^k }{ k! } \Gamma \left( \frac{k}{2} \right) \right)^2, \\
w_2(z) &= \sum_{k=1}^{\infty} \frac{ \left( \sqrt{2} z \right)^{2k - 1} }{ \left(2k - 1\right) ! } \Gamma \left( \frac{2k - 1}{2} \right) \psi \left( \frac{2k - 1}{2} \right). \\
\end{aligned}
\end{equation}
By applying \eqref{eq:ouCycleMean} and \eqref{eq:ouCycleVar} to \eqref{eq:ouStrategyRenewal}, we have the explicit formula for the expected profit per unit time and variance of profit per unit time.

\subsection{Optimization Problem}
\label{sec:pairsOptimization}

We continue to operate within the dimensionless system. For a given transaction cost $\tilde{c}$ and maximum allowed variance per unit time $\tilde{\eta}$, we find the optimal entry signal $\tilde{a}$ and exit signal $\tilde{b}$ by the optimization problem
\begin{equation}
\label{eq:ouMarkowitz}
\begin{aligned}
\max_{\tilde{a}, \tilde{b}} & & \tilde{Z}_M (\tilde{a}, \tilde{b}, \tilde{c}) \\
\text{such that} & & \tilde{Z}_V (\tilde{a}, \tilde{b}, \tilde{c}) &\leq \tilde{\eta}, \\
& & \tilde{b} &\leq \tilde{a}, \\
& & \tilde{a} &\geq 0, \\
\end{aligned}
\end{equation}
where the expected profit per unit time $\tilde{Z}_M (\tilde{a}, \tilde{b}, \tilde{c})$ and variance of profit per unit time $\tilde{Z}_V (\tilde{a}, \tilde{b}, \tilde{c})$ are given by \eqref{eq:ouReparMom}. This formulation corresponds to the strategy in which we short stock A and long stock B. The formulation for the opposite positions strategy with signals $\tilde{a}' = - \tilde{a}$ and $\tilde{b}' = - \tilde{b}$ is symmetrical. In any case, it is a nonlinear constrained optimization problem which we solve by numerical methods.

Let us denote $\tilde{a}^*$ the optimal entry signal, $\tilde{b}^*$ the optimal exit signal and $\tilde{Z}_M^*$ the optimal mean profit in the dimensionless system. Our numerical results show that the optimal exit signal is $\tilde{b}^* = - \tilde{a}^*$. This is the exactly same behavior as for the optimal exit signal in the case of unrestricted maximization of the expected profit and maximization of the Sharpe ratio as shown by \cite{Bertram2010}. This also means that the waiting part of the trading cycle for the strategy allowing for both long/short and short/long positions reduces to zero as the exit level is equal to the entry level for the strategy with opposite positions, i.e.\  $\tilde{b}^* = - \tilde{a}^* = \tilde{a}^{'*}$. The optimal strategy suggests to simply switch positions from short to long for stock A and from long to short for stock B at signal $-\tilde{a}^*$ and vice versa at signal $\tilde{a}^*$.

\subsection{Impact of Biased Estimates}
\label{sec:pairsImpact}

An inherent characteristic of the pairs trading strategy is its sensitivity to almost all aspects.  In the literature, the strategy is found to be sensitive to transaction costs, speed of execution, length of the formation period, changes in model parameters over time, diversity of traded securities and news shocks. These unpleasant properties were studied for example by \cite{Bowen2010}, \cite{Do2012}, \cite{Huck2013} and \cite{Jacobs2015a}. We add to this long list the sensitivity of the intraday pairs trading strategy to the market microstructure noise. 

We investigate the impact of biased estimates of $\tau$ and $\sigma^2$. When the market microstructure noise is not taken into account during estimation, both these parameters are overestimated. As the optimization problem \eqref{eq:ouMarkowitz} itself is formulated in the dimensionless system, it is unaffected by the values of the Ornstein--Uhlenbeck process parameters. Reparametrization \eqref{eq:ouRepar} is, however, affected. This means that the inputs to the optimization problems $\tilde{c}$ and $\tilde{\eta}$ based on the values $c$ and $\eta$ in the original parametrization can be biased. According to \eqref{eq:ouReparLev}, the transaction cost $\tilde{c}$ is biased when the ratio of $\tau$ and $\sigma^2$ is biased. The maximum allowed variance $\tilde{\eta}$ is, similarly to the variance in \eqref{eq:ouReparMom}, reparametrized as $\tilde{\eta} = 2 \eta / \sigma^2$ and is therefore biased when $\sigma^2$ is biased. A bias can also occur when the resulting optimal signals $\tilde{a}$ and $\tilde{b}$ are transformed back to $a$ and $b$ in the original parametrization. According to \eqref{eq:ouReparLev}, the entry level $a$ and exit level $b$ are biased when the ratio of $\tau$ and $\sigma^2$ is biased. The optimal mean profit per unit time $Z_M$ is also biased when either $\tau$ or $\sigma^2$ is biased according to \eqref{eq:ouReparMom}. Overall, the biased estimates of $\tau$ and $\sigma^2$ have impact on the maximum variance constraint, optimal expected profit, and optimal entry and exit signals.

We illustrate the bias of the optimal expected profit when $\sigma^2$ is correctly specified but $\tau$ is considered 10 times higher than the actual value. In this case, the maximum variance constraint is unbiased. Figure \ref{fig:ouEffFrontier} shows the efficient frontier of the mean-variance model for the optimization problem based on correctly specified as well as biased parameters. We can see that the optimization problem based on incorrectly specified parameter $\tau$ overestimates the optimal mean profit. It also finds suboptimal entry and exit signals resulting in much lower actual mean profit in comparison with the optimal mean profit based on the correct parameters.

\begin{figure}
\begin{center}
\includegraphics[width=0.8\textwidth]{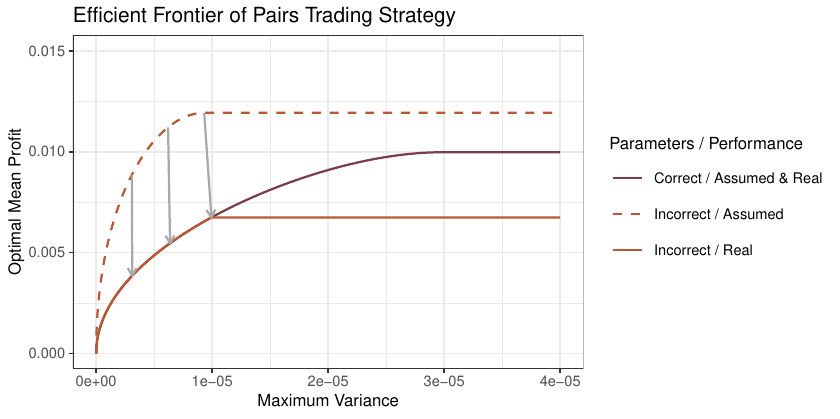}
\caption{Efficient frontier of the mean-variance model for the optimization problem based on correctly specified parameters $\mu = 1$, $\tau = 10$, $\sigma^2 = 10^{-4}$ as well as incorrect parameter $\tau = 100$.}
\label{fig:ouEffFrontier}
\end{center}
\end{figure}

\section{Application to Big Oil Companies}
\label{sec:oil}

We analyze high-frequency data of 7 Big Oil stocks traded on New York Stock Exchange (NYSE) obtained from the Daily TAQ database. Stocks of Chevron (CVX), Phillips 66 (PSX) and ExxonMobil (XOM) companies are primarily listed on NYSE while stocks of BP (BP), Eni (E), Royal Dutch Shell (RDS-A) and Total (TOT) companies are primarily listed on some other exchanges and only secondary listed on NYSE. As all 7 companies are in the same industry and they are all influenced by crude oil prices, some degree of comovement of their stock prices can be expected. The 7 considered stocks form 21 possible pairs in total. We analyze the period from January 2, 2015 to June 29, 2018 consisting of 880 trading days. The data are preprocessed using the procedure described in Appendix \ref{app:data}.

Our trading strategy utilizes results from Sections \ref{sec:ou} and \ref{sec:pairs}. First, we analyze historical intraday data. We separately estimate the parameters of the Ornstein--Uhlenbeck process for each considered pair on each considered day. Some days exhibit strong mean-reversion suggesting the Ornstein--Uhlenbeck process with high speed of reversion as illustrated in the upper plot of Figure \ref{fig:ouDay} while others exhibit random walk behaviour suggesting the Wiener process as illustrated in the lower plot of Figure \ref{fig:ouDay}. Days with high speed of reversion and high volatility offer more opportunities for profit. Second, we utilize time series models to capture time-varying nature of daily parameter values. This allows us to predict future parameter values. Third, assuming the Ornstein--Uhlenbeck process with specific parameters, we find the optimal entry and exit signals together with the expected profit and the variance of the profit for a given pair on a given day. Based on the values of the mean profit and its variance, we decide whether to trade the given pair on the given day or not. If the decision is positive, the trading is then controlled by the optimal entry and exit signals.

\begin{figure}
\begin{center}
\includegraphics[width=0.8\textwidth]{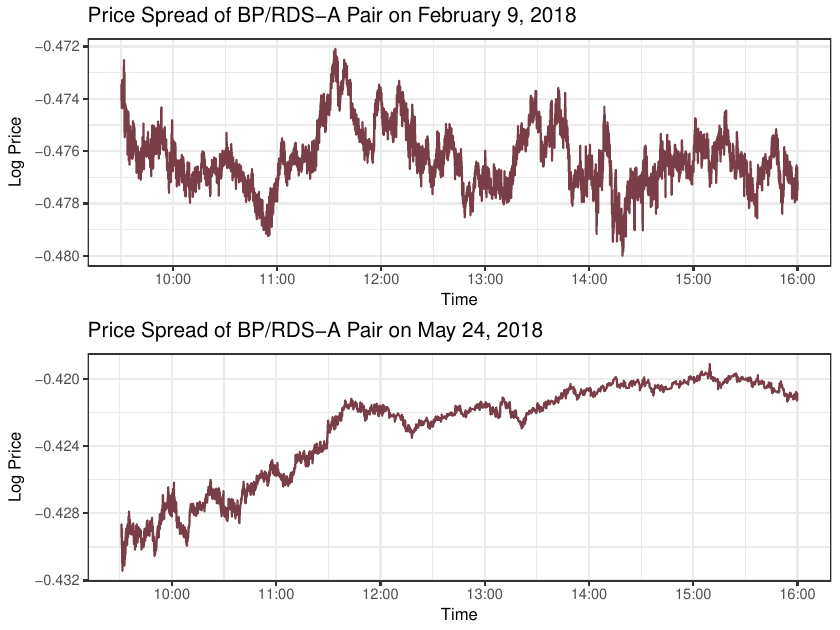}
\caption{Price spread of BP/RDS-A pair resembling Ornstein--Uhlenbeck process on February 9, 2018 and Wiener process on May 24, 2018.}
\label{fig:ouDay}
\end{center}
\end{figure}

\subsection{Estimators Performance}
\label{sec:oilEstimators}

We compare the noise-sensitive maximum likelihood estimator (TICK-MLE) with the noise-robust maximum likelihood estimator (TICK-MLE-NR) based on tick data. The first question is whether the market microstructure noise is indeed present in the observed prices. As the high-frequency data studies agree that the noise is present (see e.g.\  \citealp{Hansen2006}), we address this issue only briefly using a graphical analysis. In Figure \ref{fig:ouSignature}, we adopt the so-called volatility signature plot. The plot shows the dependence of the average estimated value of variance on the sampling interval. For tick data, the sampling interval $k$ refers to data consisting of each $k$-th observation. For example, value 1 corresponds to complete tick data while value 2 corresponds to every second observation being dropped. The number of observations for sampling interval $k$ is approximately $n/k$, where $n$ is the number of observations of complete tick data. We can see in Figure \ref{fig:ouSignature} that the variance estimated by the noise-sensitive method increases with the number of observations. This is exactly the behavior caused by the market microstructure noise. Noise-robust estimator, on the other hand, sticks around a constant value. For $k=1$, the bias of the TICK-MLE method is quite big causing very distorted image of the price volatility.

\begin{figure}
\begin{center}
\includegraphics[width=0.8\textwidth]{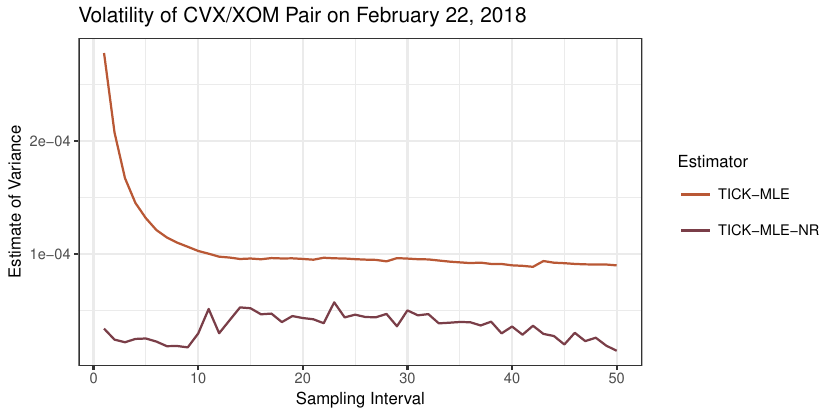}
\caption{Volatility signature plot of CVX/XOM pair on February 22, 2018.}
\label{fig:ouSignature}
\end{center}
\end{figure}

The second question about the market microstructure noise is whether the independent white noise assumption is met in practice. \cite{Hansen2006} analyze stocks traded on the NYSE and NASDAQ exchanges and find that the market microstructure noise present in prices is dependent in time and dependent on efficient prices. Using volatility signature plots, they notice decreasing volatility with increasing number of observations $n/k$, which can be explained only by the innovations in the noise process negatively correlated with the efficient returns. When we analyze stock prices, we achieve the same results. However, when we analyze spreads between pairs of stocks, the volatility estimated by the noise-sensitive method is distinctly increasing with shorter sampling interval in the vast majority of days as shown in the example in Figure \ref{fig:ouSignature}. We argue that the noise in the spread process has twice as many sources than the noise in a price process which diminishes dependency of the noise. For this reason, we consider the white noise assumption reasonable for the pair spread process, even when it is not suitable for the price process itself.

The average parameters estimated by the TICK-MLE and TICK-MLE-NR methods for each pair are reported in Table \ref{tab:ouParameters}. The estimated means $\mu$ are quite similar for the two methods while parameters $\tau$ and $\sigma$ are much higher for the TICK-MLE method. On average, the speed of reversion $\tau$ is 6.36 times higher and the standard error $\sigma$ is 2.21 higher (the variance $\sigma^2$ is 4.78 higher) when estimated by the TICK-MLE method. Note that Table \ref{tab:ouParameters} reports standard deviation $\sigma$ and not variance $\sigma^2$. Following our theory and Figure \ref{fig:ouSignature}, we argue that the estimates of $\tau$ and $\sigma$ by the TICK-MLE method are significantly biased and this estimator should be avoided. The proposed TICK-MLE-NR method, on the other hand, is not affected by the noise while utilizing all available tick data.

\begin{table}
\centering
\resizebox{\textwidth}{!}{
\begin{tabular}{lrrrrrrrr}
\toprule
& \multicolumn{3}{c}{TICK-MLE} & \multicolumn{3}{c}{TICK-MLE-NR} \\
\cmidrule(l{3pt}r{3pt}){2-4} \cmidrule(l{3pt}r{3pt}){5-7}
Pair        &  $\mu$  &  $\tau$ &               $\sigma$ &   $\mu$ &  $\tau$ &               $\sigma$ \\
\midrule
BP / CVX    & -1.0546 & 29.9279 & $1.7401 \cdot 10^{-2}$ & -1.0552 &  3.9524 & $0.6676 \cdot 10^{-2}$ \\
BP / E      &  0.1149 & 35.7787 & $1.7458 \cdot 10^{-2}$ &  0.1149 &  6.5204 & $0.9428 \cdot 10^{-2}$ \\
BP / PSX    & -0.8416 & 25.5638 & $1.9848 \cdot 10^{-2}$ & -0.8418 &  4.1090 & $0.8054 \cdot 10^{-2}$ \\
BP / RDS-A  & -0.4256 & 65.3009 & $1.7258 \cdot 10^{-2}$ & -0.4256 &  6.4025 & $0.6903 \cdot 10^{-2}$ \\
BP / TOT    & -0.3324 & 49.7040 & $1.8318 \cdot 10^{-2}$ & -0.3326 &  7.0370 & $0.8302 \cdot 10^{-2}$ \\
BP / XOM    & -0.8230 & 26.4986 & $1.5536 \cdot 10^{-2}$ & -0.8240 &  3.0105 & $0.6908 \cdot 10^{-2}$ \\
CVX / E     &  1.1697 & 18.0115 & $1.6074 \cdot 10^{-2}$ &  1.1690 &  4.4924 & $0.8804 \cdot 10^{-2}$ \\
CVX / PSX   &  0.2129 & 21.9259 & $1.8796 \cdot 10^{-2}$ &  0.2130 &  3.9646 & $0.8357 \cdot 10^{-2}$ \\
CVX / RDS-A &  0.6286 & 28.4128 & $1.7110 \cdot 10^{-2}$ &  0.6284 &  3.7525 & $0.6507 \cdot 10^{-2}$ \\
CVX / TOT   &  0.7223 & 24.6909 & $1.7381 \cdot 10^{-2}$ &  0.7233 &  4.8954 & $0.8190 \cdot 10^{-2}$ \\
CVX / XOM   &  0.2316 & 33.6227 & $1.5280 \cdot 10^{-2}$ &  0.2316 &  3.2192 & $0.5324 \cdot 10^{-2}$ \\
E / PSX     & -0.9566 & 19.8339 & $1.9473 \cdot 10^{-2}$ & -0.9565 &  5.6958 & $1.1065 \cdot 10^{-2}$ \\
E / RDS-A   & -0.5393 & 32.5781 & $1.6314 \cdot 10^{-2}$ & -0.5396 &  5.9070 & $0.8110 \cdot 10^{-2}$ \\
E / TOT     & -0.4473 & 56.8259 & $1.9806 \cdot 10^{-2}$ & -0.4472 & 10.1327 & $0.9431 \cdot 10^{-2}$ \\
E / XOM     & -0.9389 & 15.2862 & $1.3825 \cdot 10^{-2}$ & -0.9388 &  3.4012 & $0.7846 \cdot 10^{-2}$ \\
PSX / RDS-A &  0.4160 & 22.8768 & $1.9327 \cdot 10^{-2}$ &  0.4161 &  3.6241 & $0.8121 \cdot 10^{-2}$ \\
PSX / TOT   &  0.5094 & 23.1866 & $2.0295 \cdot 10^{-2}$ &  0.5096 &  5.5541 & $1.0106 \cdot 10^{-2}$ \\
PSX / XOM   &  0.0186 & 19.4030 & $1.7160 \cdot 10^{-2}$ &  0.0186 &  2.8399 & $0.7516 \cdot 10^{-2}$ \\
RDS-A / TOT &  0.0933 & 51.7991 & $1.7501 \cdot 10^{-2}$ &  0.0934 &  7.1039 & $0.6976 \cdot 10^{-2}$ \\
RDS-A / XOM & -0.3985 & 25.8139 & $1.5277 \cdot 10^{-2}$ & -0.3987 &  2.7363 & $0.5473 \cdot 10^{-2}$ \\
TOT / XOM   & -0.4907 & 21.0923 & $1.5298 \cdot 10^{-2}$ & -0.4910 &  3.5759 & $0.7001 \cdot 10^{-2}$ \\
\midrule
Average     & -0.1491 & 30.8635 & $1.7368 \cdot 10^{-2}$ & -0.1492 &  4.8537 & $0.7862 \cdot 10^{-2}$ \\
\bottomrule
\end{tabular}
}
\caption{Average values of the Ornstein--Uhlenbeck process parameters estimated by the noise-sensitive and noise-robust estimators from January 2, 2015 to June 29, 2018.}
\label{tab:ouParameters}
\end{table}

\subsection{Time-Varying Parameters}
\label{sec:oilTime}

In this section, we present the time series models used for time-varying parameters of the Ornstein--Uhlenbeck process. We assume values of parameters can change on each day $i = 1,\ldots,h$. In other words, we assume the time-varying parameters to follow piecewise constant process, in which parameters are constant during the whole day. For each parameter, we consider separate model. The main purpose of these models is to forecast future values of the parameters.

Daily mean parameter $\mu_i$ is modeled as the AR(1) process with the opening price $X_{0,i}$ on day $i$ as an exogenous variable, i.e.
\begin{equation}
\label{eq:ouTimeModelMu}
\mu_i = a + b \mu_{i-1} + c X_{0,i} + \varepsilon_i, \qquad i=1,\ldots,h,
\end{equation}
where $a$, $b$, $c$ are the coefficients and $\varepsilon_i$ is the Gaussian white noise. This is a very similar idea to the doubly mean-reverting process of \cite{Liu2017}. In their study, they consider the prices to follow two mean-reverting processes on two frequencies. The low frequency corresponds to daily opening and closing prices while the high frequency corresponds to intraday prices. In our case, the low frequency mean-reverting process is represented by the autoregressive process for the daily mean parameter.

Daily speed of reversion parameter $\tau_i$ is modeled only by the mean value, i.e.
\begin{equation}
\label{eq:ouTimeModelTau}
\tau_i = a' + \varepsilon_i', \qquad i=1,\ldots,h,
\end{equation}
where $a'$ is the coefficient and $\varepsilon_i'$ is the Gaussian white noise. The one-step-ahead forecast of $\tau_i$ is then simply the average of its past values. We resort to this static model as we find no autocorrelation structure in the empirical study.

For the daily variance parameter $\sigma^2_i$, we utilize the HAR model of \cite{Corsi2009}. They model volatility by the realized variance over different time periods. Specifically, the daily realized variance is dependent on the realized variance of the previous day, the previous week, and the previous month. In our case, the logarithm of the parameter $\sigma^2_i$ follows the autoregressive process
\begin{equation}
\label{eq:ouTimeModelSigma}
\ln \sigma_i^2 = a'' + b'' \ln \sigma_{i-1}^2 + c'' \frac{1}{5} \sum_{j=1}^{5} \ln \sigma_{i-j}^2 + d'' \frac{1}{22} \sum_{j=1}^{22} \ln \sigma_{i-j}^2 + \varepsilon_i'', \qquad i=1,\ldots,h,
\end{equation}
where $a''$, $b''$, $c''$, $d''$ are the coefficients and $\varepsilon_i''$ is the Gaussian white noise.

We train the models using a rolling window of 132 days (approximately 6 months) and perform one-step-ahead forecasts. We find that the model \eqref{eq:ouTimeModelMu} for the long-term mean parameter explains 96 \% of the variance of $\mu_i$ on average while the model \eqref{eq:ouTimeModelSigma} for the variance parameter explains 25 \% of the variance of $\sigma^2_i$ on average. By definition, the model \eqref{eq:ouTimeModelTau} for speed of reversion parameter explains exactly 0 \% of the variance of $\tau$. Overall, we find that the models for $\mu_i$ and $\sigma_i^2$ parameters are satisfactory while the parameter $\tau_i$ is very hard to predict.

\subsection{Trading Algorithm}
\label{sec:oilAlgo}

For a set of parameters of the Ornstein--Uhlenbeck process obtained by the forecasting models and a given maximum allowed variance of the profit $\eta$, we find the optimal entry and exit signals together with the maximal expected profit. As the forecasted parameter values are uncertain, we trade only if the expected profit is larger than a given threshold $\zeta$.

We use transaction costs $c = 0.0015$ per round-trip pair-trade. In the literature, this is considered as a moderate level of transaction costs. For example, \cite{Avellaneda2010}, \cite{Bertram2010} and \cite{Liu2017} use an optimistic transaction costs level of 0.0010, \cite{Bowen2010} use a moderate level of 0.0015 and \cite{Bogomolov2013} uses a conservative level of 0.0040.

We summarize the proposed pairs trading strategy with notes regarding our specific setting. First, we need to select several parameters of the strategy. The initialization of the strategy requires the following steps:
\begin{enumerate}
\item A set of potentially tradable pairs is selected. The number of pairs is denoted as $p$. In our case, we consider $p=21$ pairs created from 7 stocks.
\item The length of history $h$ is selected. In our case, we use history of $h=132$ days corresponding roughly to 6 months.
\item The maximum allowed variance $\eta$ for daily profit is selected. In our case, we consider $\eta = 10^{-5}$, $\eta = 5 \cdot 10^{-5}$, and $\eta = \infty$. The value $\eta = 5 \cdot 10^{-5}$ is found to yield the best results.
\item The minimum allowed mean $\zeta$ for daily profit is selected. In our case, we consider $\zeta \in (0, 0.7)$. The value 0.009 is found to yield the best results.
\end{enumerate}

Next, we describe our strategy for a single trading day $h+1$. The execution lies in the following steps:
\begin{enumerate}
\item Historical intraday data are analyzed using the methodology described in Section \ref{sec:ou}. For each pair $j=1,\ldots,p$ and each historical day $i=1,\ldots,h$, the Ornstein--Uhlenbeck parameters $\mu_{j,i}$, $\tau_{j,i}$ and $\sigma_{j,i}^2$ are estimated. In our case, we consider the TICK-MLE and TICK-MLE-NR estimators.
\item Time series models described in Section \ref{sec:oilTime} are utilized to capture time-varying nature of daily parameter values and predict their future values. For each pair $j=1,\ldots,p$, the models \eqref{eq:ouTimeModelMu}, \eqref{eq:ouTimeModelTau} and \eqref{eq:ouTimeModelSigma} for daily Ornstein--Uhlenbeck parameters are estimated using history $h$. Future parameter values $\mu_{j,h+1}$, $\tau_{j,h+1}$ and $\sigma_{j,h+1}^2$ are then forecasted. Prices during a single day in future are then assumed to follow the Ornstein--Uhlenbeck process with forecasted parameters.
\item The optimal strategy described in Section \ref{sec:pairs} is determined. For each pair $j=1,\ldots,p$, the optimal entry signal $a_j^*$, the optimal exit signal $b_j^*$ and the optimal mean profit $Z_{M,j}^*$ are found using \eqref{eq:ouMarkowitz}. In this model, the mean profit $Z_{M,j}$ is maximized while the variance of the profit $Z_{V,j}$ is lower than $\eta$. For the opposite pairs trade, the optimal entry signal is $a_j^{'*} = b_j^*$, the optimal exit signal is $b_j^{'*} = a_j^*$ and the optimal mean profit is $Z_{M,j}^{'*} = Z_{M,j}^*$.
\item For each pair $j=1,\ldots,p$, it is decided whether this pair will be traded on day $h+1$ or not. The pair will be traded if its optimal mean is higher than the selected threshold, i.e.\  $Z_{M,j}^* \geq \zeta$.
\item For each tradable pair $j$, intraday prices are monitored. When the price reaches the entry level $a_j^*$ or $a_j^{'*}$, the appropriate pairs trade is entered as described in Section \ref{sec:pairs}. When the price reaches the exit level $b_j^* = a_j^{'*}$ or $b_j^{'*} = a_j^{*}$, long and short positions are switched. Right before the market closes, both positions are closed regardless the price.
\end{enumerate}

Closing positions before market close helps mitigate overnight risk. To further manage potential losses and prevent significant divergences, additional risk management tools, such as stop-loss mechanisms, can be implemented.

\subsection{Strategy Performance}
\label{sec:oilStrategy}

We asses the profitability of the pairs trading strategy for the 21 pairs comprising of the 7 Big Oil companies. As we use 6 months history for the training of the forecasting models, we evaluate the strategy from the second half of the year 2015 to the second half of the year 2018.

We consider $\eta = 10^{-5}$, $\eta = 5 \cdot 10^{-5}$, and $\eta = \infty$ as levels for the maximum allowed variance. Figure \ref{fig:ouStrategyProfit} shows the total daily profit of the strategy based on 21 pairs for various values of the minimum mean profit $\zeta$. We can see that the profit is quite sensitive to the selection of thresholds $\eta$ and $\zeta$. When the expected mean is not limited, almost all pairs are traded on all days resulting in a huge loss. When the minimum mean profit $\zeta$ is set around 0.009, the strategy based on the TICK-MLE-NR estimator performs the best and achieves daily profit up to 0.0069 in terms of the continuous compound rate of return for $\eta = 5 \cdot 10^{-5}$. When we further increase the threshold for minimum mean profit $\zeta$, less trades are carried out and even the profitable trades are cut resulting in decline of the profit. Naturally, the profit converges to zero with increasing minimum mean profit $\zeta$.

\begin{figure}
\begin{center}
\includegraphics[width=0.8\textwidth]{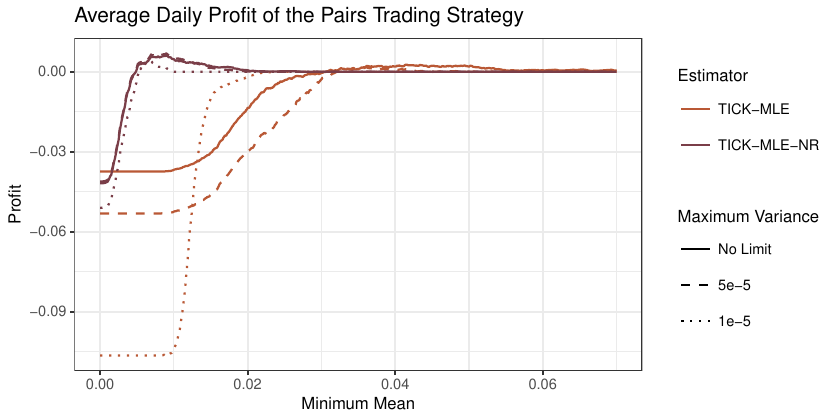}
\caption{Dependence of the daily profit on the maximum variance $\eta$ and minimum mean $\zeta$ for the noise-sensitive and noise-robust estimators from January 2, 2015 to June 29, 2018.}
\label{fig:ouStrategyProfit}
\end{center}
\end{figure}

Interestingly, the number of trades and the profit are not evenly distributed throughout the years. Most trades are executed during the years 2015, 2016 and 2018 while the year 2017 is quiet period for the strategy based on the TICK-MLE-NR estimator. We attribute this to the lower volatility of the spread prices during 2017.

Table \ref{tab:ouStrategyProfit} reports, for each pair separately, the average daily number of trades, the percentage of days with zero or positive profit, the average daily profit, and the daily Sharpe ratio (assuming a zero risk-free rate). Generally, pairs with higher estimated values of $\tau$ and $\sigma^2$ are traded more as their expected profit is also higher. We focus on the TICK-MLE-NR estimator with the most profitable setting of the maximum variance of the profit $5 \cdot 10^{-5}$ and the minimum mean profit $0.009$. Table \ref{tab:ouStrategyProfit} indicates that E/PSX, E/TOT and PSX/TOT are the most traded pairs while E/PSX and E/TOT are also the most profitable pairs. Table \ref{tab:ouParameters} shows that these pairs have the above average estimated values of $\tau$ and $\sigma^2$. Overall, the strategy based on the TICK-MLE-NR estimator has a daily Sharpe ratio of 0.12 and an annualized Sharpe ratio of 1.97, which is a decent value.

\begin{table}
\centering
\resizebox{\textwidth}{!}{
\begin{tabular}{lrrrrrrrr}
\toprule
& \multicolumn{4}{c}{TICK-MLE} & \multicolumn{4}{c}{TICK-MLE-NR} \\
\cmidrule(l{3pt}r{3pt}){2-5} \cmidrule(l{3pt}r{3pt}){6-9}
Pair       & Trades & Hit \% & Profit & Sharpe & Trades & Hit \% & Profit & Sharpe \\
\midrule
  BP / CVX & 5.59 & 42.56 & -0.0032 & -0.20 & 0.14 & 99.04 & 0.0000 & 0.00 \\ 
  BP / E & 5.89 & 46.97 & 0.0009 & 0.06 & 0.87 & 94.21 & 0.0011 & 0.13 \\ 
  BP / PSX & 5.99 & 44.08 & -0.0044 & -0.24 & 0.25 & 97.80 & -0.0000 & -0.00 \\ 
  BP / RDSA & 4.57 & 36.83 & -0.0024 & -0.23 & 0.03 & 99.44 & -0.0001 & -0.04 \\ 
  BP / TOT & 5.24 & 40.91 & -0.0020 & -0.17 & 0.25 & 97.66 & 0.0002 & 0.06 \\ 
  BP / XOM & 4.83 & 38.84 & -0.0043 & -0.29 & 0.01 & 99.86 & -0.0001 & -0.04 \\ 
  CVX / E & 5.93 & 46.42 & -0.0009 & -0.05 & 0.76 & 95.87 & 0.0008 & 0.06 \\ 
  CVX / PSX & 6.05 & 39.94 & -0.0053 & -0.27 & 0.34 & 95.45 & -0.0008 & -0.09 \\ 
  CVX / RDSA & 5.25 & 41.88 & -0.0035 & -0.23 & 0.12 & 99.44 & 0.0001 & 0.03 \\ 
  CVX / TOT & 6.35 & 45.04 & -0.0021 & -0.12 & 0.65 & 95.87 & 0.0006 & 0.06 \\ 
  CVX / XOM & 4.61 & 38.29 & -0.0048 & -0.34 & 0.03 & 100.00 & 0.0001 & 0.04 \\ 
  E / PSX & 6.79 & 45.18 & -0.0016 & -0.07 & 2.21 & 83.75 & 0.0015 & 0.08 \\ 
  E / RDSA & 5.44 & 45.24 & 0.0001 & 0.01 & 0.71 & 96.08 & 0.0012 & 0.12 \\ 
  E / TOT & 6.28 & 48.48 & 0.0019 & 0.12 & 1.52 & 88.71 & 0.0018 & 0.16 \\ 
  E / XOM & 4.89 & 48.07 & -0.0001 & -0.01 & 0.12 & 98.62 & -0.0000 & -0.00 \\ 
  PSX / RDSA & 5.84 & 42.58 & -0.0048 & -0.26 & 0.25 & 98.46 & 0.0001 & 0.02 \\ 
  PSX / TOT & 6.42 & 45.04 & -0.0037 & -0.19 & 1.39 & 87.33 & -0.0000 & -0.00 \\ 
  PSX / XOM & 5.47 & 40.50 & -0.0052 & -0.28 & 0.10 & 99.31 & 0.0000 & 0.01 \\ 
  RDSA / TOT & 5.18 & 44.54 & -0.0013 & -0.11 & 0.22 & 97.90 & 0.0001 & 0.03 \\ 
  RDSA / XOM & 4.66 & 42.44 & -0.0041 & -0.29 & 0.00 & 100.00 & 0.0000 & 0.00 \\ 
  TOT / XOM & 5.06 & 43.11 & -0.0025 & -0.15 & 0.02 & 99.86 & 0.0001 & 0.04 \\ 
\midrule
Total & 115.84 & 33.75 & -0.0530 & -0.36 & 9.96 & 83.88 & 0.0069 & 0.12 \\ 
\bottomrule
\end{tabular}
}
\caption{Average daily number of trades, percentage of days with zero or positive profit, average daily profit, and daily Sharpe ratio, with $\eta = 5\cdot10^{-5}$ and $\zeta = 0.009$ for the noise-sensitive and noise-robust estimators from January 2, 2015 to June 29, 2018.}
\label{tab:ouStrategyProfit}
\end{table}

Finally, we compare the TICK-MLE and TICK-MLE-NR estimators. Figure \ref{fig:ouStrategyProfit} illustrates that both estimators have quite different ideas of the mean profit and its variance. As shown in Section \ref{sec:pairsImpact}, the values of the moments are quite distorted when the parameter estimates are biased as they are in the case of the TICK-MLE estimator. More important, even when selecting the best thresholds for the minimum mean profit $\zeta$ and the maximum variance of the profit $\eta$ for each method separately, the TICK-MLE-NR estimator significantly outperforms the TICK-MLE estimator. This is because the optimization based on the TICK-MLE estimator finds suboptimal values of entry and exit signals. The TICK-MLE-NR estimator, on the other hand, finds optimal values leading to a much greater profit. This finding is the key result of our pairs trading application.

In Appendix \ref{app:new}, we apply our strategy to a supplementary empirical sample from March 1, 2022, to December 29, 2023. This period follows Russia's invasion of Ukraine in late February 2022, which shocked the oil markets. It is found that our strategy is not profitable during this period of market dislocations. However, the empirical differences between the TICK-MLE and TICK-MLE-NR estimators are still present. The performance of the strategy may thus change over time, but properly addressing market microstructure noise remains relevant.

\section{Conclusion}
\label{sec:conclusion}

Our paper has two main contributions:

\begin{enumerate}
\item We propose three different estimators of the Ornstein--Uhlenbeck process which directly take the market microstructure noise into account. For initial estimates, we recommend the closed-form method of moments. For regularly spaced observations, we introduce an approach based on the reparametrization of the process to the ARMA(1,1) process and subsequent estimation by the maximum likelihood or conditional sum-of-squares methods. For irregularly spaced observations, we suggest a method based on the maximum likelihood. We show in the simulation study as well as in the empirical study that the proposed noise-robust estimators outperform the traditional estimators ignoring the noise. The behavior of the estimators is consistent with the high-frequency literature dealing with the market microstructure noise.
\item We propose a novel pairs trading strategy that finds the entry and exit signals using mean-variance optimization. Adding a suitable risk constraint to the optimization problem reduces losses due to uncertainty of the out-of-sample Ornstein--Uhlenbeck parameters. In the empirical study, we show that this strategy can be viable when the market microstructure noise is taken into account and the proposed estimators are used. The proposed mean-variance optimization is further studied by \cite{Holy2022c}.
\end{enumerate}

Although we restrict ourselves to the pairs trading strategy in the second part of the paper, our findings that the traditional estimators are biased when the market microstructure noise is present are general. Our proposed noise-robust estimators of the Ornstein--Uhlenbeck can be used in various financial applications in which time series exhibit mean-reverting behavior such as modeling of currency exchange rates, commodity prices, interest rates, and stochastic volatility of financial assets.

Possible directions for future research include modeling jumps in the price process and further analyzing parameter uncertainty in the optimization of trading signals.

\section*{Acknowledgements}
\label{sec:acknow}

We would like to thank Michal \v{C}ern\'{y}, Tom\'{a}\v{s} Cipra, and Alena Hol\'{a} for their comments. We would also like to thank participants of the 3rd Conference and Workshop on Statistical Methods in Finance, Chennai, December 16--19, 2017, the 20th Winter Workshop ROBUST, Rybn\'{i}k, January 21--26, 2018 and the 30th European Conference on Operational Research, Dublin, June 23--26, 2019 for fruitful discussions.

\section*{Funding}
\label{sec:fund}

The work on this paper was supported by the Internal Grant Agency of the Prague University of Economics and Business under project F4/63/2016 and the Czech Science Foundation under project 25-18028S.

\appendix

\section{Conditional Distribution with the Noise}
\label{app:conditional}

We show that the conditional probability distribution of the Ornstein--Uhlenbeck process contaminated by the Gaussian white noise is the normal distribution with moments given by \eqref{eq:ouObsMomentsCon}. For this purpose, we utilize the following proposition with $P = P_0$, $\mu_P = \mu$, $\sigma^2_P = \sigma^2 / (2 \tau)$, $E = E_0$, $\mu_E = 0$, $\sigma^2_E = \omega^2$ and $X = X_0$.

\begin{proposition}
\label{th:normalCondNoisy}
Let $P \sim \text{N}(\mu_P, \sigma^2_P)$, $E \sim \text{N}(\mu_E, \sigma^2_E)$ and $P \perp E$. Let $X=P+E$. The conditional probability density function is then
\begin{equation}
f_{P}(p \mid X=x) = \frac{1}{\sqrt{2\pi \sigma^2_C(x)}} \exp \left\{ - \frac{\left(p-\mu_C(x)\right)^2}{2\sigma^2_C(x)} \right\},
\end{equation}
where
\begin{equation}
\begin{aligned}
\mu_C(x) &= \frac{\mu_P\sigma^2_E-\mu_E\sigma^2_P+x\sigma^2_P}{\sigma^2_P+\sigma^2_E}, \\
\sigma^2_C(x) &= \frac{\sigma^2_P\sigma^2_E}{\sigma^2_P+\sigma^2_E}.
\end{aligned}
\end{equation}
\end{proposition}

\begin{proof}
The joint probability density function of $P$ and $X$ is given by
\begin{equation}
\begin{aligned}
g_{P,X}(p,x) &= \frac{1}{\sqrt{2\pi\sigma^2_P}}\exp \left\{ -\frac{(p-\mu_P)^2}{2\sigma^2_P} \right\} \frac{1}{\sqrt{2\pi\sigma^2_E}}\exp \left\{ -\frac{(x-p-\mu_E)^2}{2\sigma^2_E} \right\} \\
&= \frac{1}{\sqrt{2\pi\sigma^2_P}\sqrt{2\pi\sigma^2_E}}\exp \left\{ -\frac{\sigma^2_P+\sigma^2_E}{2\sigma^2_P\sigma^2_E}p^2+ \frac{\mu_P\sigma^2_E+x\sigma^2_P-\mu_E\sigma^2_P}{\sigma^2_P\sigma^2_E}p \right. \\
&\quad \left. + \frac{2x\mu_E\sigma^2_P-\mu_P^2\sigma^2_E- x^2\sigma^2_P -\mu^2_E\sigma^2_P}{2\sigma^2_P\sigma^2_E} \right\}.
\end{aligned}
\end{equation}
Using the property of Gaussian function integral
\begin{equation}
\int_{-\infty}^\infty \exp \left\{ -ap^2+bp+c \right\}\mathrm{d}p = \sqrt{\frac{\pi}{a}}\exp \left\{ \frac{b^2}{4a}+c \right\},
\end{equation}
we get the marginal probability density function
\begin{equation}
\begin{aligned}
h_{X}(x) &= \int_{-\infty}^\infty g_{P,X}(p,x) \mathrm{d}p \\
&= \frac{1}{\sqrt{2\pi\sigma^2_P}\sqrt{2\pi\sigma^2_E}}\sqrt{\frac{\pi}{\frac{\sigma^2_P+\sigma^2_E}{2\sigma^2_P\sigma^2_E}}} \\
& \quad \times \exp \left\{ \frac{\left(\frac{\mu_P\sigma^2_E+x\sigma^2_P-\mu_E\sigma^2_P}{\sigma^2_P\sigma^2_E}\right)^2}{4\left( \frac{\sigma^2_P+\sigma^2_E}{2\sigma^2_P\sigma^2_E} \right)} + \frac{2x\mu_E\sigma^2_P-\mu_P^2\sigma^2_E- x^2\sigma^2_P -\mu^2_E\sigma^2_P}{2\sigma^2_P\sigma^2_E} \right\} \\
&= \frac{1}{\sqrt{2\pi \left( \sigma^2_P + \sigma^2_E \right)}} \exp \left\{ - \frac{(\mu_P-x+\mu_E)^2}{2(\sigma^2_P+\sigma^2_E)} \right\}.
\end{aligned}
\end{equation}
The conditional probability density function is then derived as
\begin{equation}
\begin{aligned}
f_{P}(p \mid X=x) &= \frac{g_{P,X}(p,x)}{h_{X}(x)} \\
&= \frac{1}{\sqrt{2\pi \frac{\sigma^2_P\sigma^2_E}{\sigma^2_P+\sigma^2_E}}} \exp \left\{ - \frac{(p-\mu_P)^2}{2\sigma^2_P} - \frac{(x-p-\mu_E)^2}{2\sigma^2_E}  + \frac{(\mu_P-x+\mu_E)^2}{2(\sigma^2_P+\sigma^2_E)} \right\} \\
&= \frac{1}{\sqrt{2\pi \sigma^2_C(x)}} \exp \left\{ - \frac{\left(p-\mu_C(x)\right)^2}{2\sigma^2_C(x)} \right\}.
\end{aligned}
\end{equation}
\end{proof}

\section{Method of Moments Without the Noise}
\label{app:moments}

We derive the traditional method of moments for the case of the equidistantly sampled Ornstein--Uhlenbeck process with no noise. As we need to estimate parameters $\mu$, $\tau$ and $\sigma$, we utilize three unconditional moments $\textrm{E} [ P_{T_i} ]$, $\textrm{var} [ P_{T_i} ]$, and $\textrm{cov} [ P_{T_i}, P_{T_{i-1}} ]$ presented in \eqref{eq:ouMomentsUnc}. We can estimate these moments using observed values $p_{T_0}, p_{T_1}, \ldots, p_{T_n}$ as
\begin{equation}
\begin{aligned}
M_{1,n} &= \frac{1}{n+1} \sum_{i=0}^n p_{T_i}, \\
M_{2,n} &= \frac{1}{n} \sum_{i=0}^n (p_{T_i} - M_{1,n})^2, \\
M_{3,n} &= \frac{1}{n-1} \sum_{i=1}^n (p_{T_i} - M_{1,n})(p_{T_{i-1}} - M_{1,n}), \\
\end{aligned}
\end{equation}
By solving equations
\begin{equation}
\textrm{E} [ P_{T_i} ] = M_{1,n}, \quad \textrm{var} [ P_{T_i} ] = M_{2,n}, \quad \textrm{cov} [ P_{T_i}, P_{T_{i-1}} ] = M_{3,n},
\end{equation}
we get estimates
\begin{equation}
\begin{aligned}
\hat{\mu} &= M_{1,n}, \\
\hat{\tau} &= n \log \frac{M_{2,n}}{M_{3,n}}, \\
\hat{\sigma}^2 &= 2 n M_{2,n} \log \frac{M_{2,n}}{M_{3,n}}. \\
\end{aligned}
\end{equation}

We illustrate the bias of the method of moments when the Ornstein--Uhlenbeck process is contaminated by the white noise with standard deviation $\omega$. Parameter $\mu$ can be consistently estimated by sample mean. For the other two parameters, the situation is more difficult. Parameter $\tau$ can be estimated using equation
\begin{equation}
\label{eq:ouMomTauEst}
\begin{aligned}
\tau_{P,n}  &= n \log \frac{\mathrm{var} [ P_{T_{i-1}} ]}{\mathrm{cov} [ P_{T_{i}}, X_{T_{i-1}} ]} \\
&= - n \log \mathrm{cor} [ P_{T_i}, P_{T_{i-1}} ].
\end{aligned}
\end{equation}
The method of moments replaces the theoretical correlation in this equation by the sample correlation to estimate $\tau$. However, if the actual process follows $X_i$, the equality \eqref{eq:ouMomTauEst} does not hold and instead we have
\begin{equation}
\label{eq:ouMomTauBias}
\begin{aligned}
\tau_{X,n} &= n \log \frac{\mathrm{var} [ X_{i-1} ]}{\mathrm{cov} [ X_i, X_{i-1} ]} \\
&= - n \log \mathrm{cor} [ X_i, X_{i-1} ] \\
&= - n \log \left( \frac{\sigma^2}{\sigma^2 + 2\tau \omega^2} e^{-\tau (T_{i} - T_{i-1})} \right) \\
&= \tau_{P,n} - n \log \frac{\sigma^2}{\sigma^2 + 2 \tau \omega^2}.
\end{aligned}
\end{equation}
The estimate $\tau_{X,n}$ is a function of the number of observations, which for $n \to \infty$ linearly diverges to infinity. Similarly, parameter $\sigma$ can be estimated using equation
\begin{equation}
\label{eq:ouMomSigmaEst}
\begin{aligned}
\sigma^2_{P,n} &= 2 n \mathrm{var} [ P_{T_i} ] \log \frac{\mathrm{var} [ P_{T_{i-1}} ]}{\mathrm{cov} [ P_{T_{i}}, X_{T_{i-1}} ]} \\
&= - 2 n \mathrm{var} [ P_{T_i} ] \log \mathrm{cor} [ P_{T_i}, P_{T_{i-1}} ].
\end{aligned}
\end{equation}
When the process is noisy, we have
\begin{equation}
\label{eq:ouMomSigmaBias}
\begin{aligned}
\sigma^2_{X,n} &= 2 n \mathrm{var} [ X_i ] \log \frac{\mathrm{var} [ X_{i-1} ]}{\mathrm{cov} [ X_i, X_{i-1} ]} \\
&= - 2 n \mathrm{var} [ X_i ] \log \mathrm{cor} [X_i, X_{i-1} ] \\
&= - 2 n \left( \frac{\sigma^2}{2\tau} + \omega^2 \right) \log \left( \frac{\sigma^2}{\sigma^2 + 2\tau \omega^2} e^{-\tau (T_{i} - T_{i-1})} \right) \\
&= \sigma^2_{P,n} + 2 \tau \omega^2 - 2 n \left( \frac{\sigma^2}{2 \tau} + \omega^2 \right) \log \frac{\sigma^2}{\sigma^2 + 2 \tau \omega^2},
\end{aligned}
\end{equation}
which also linearly diverges to infinity for $n \to \infty$. We show the bias of $\tau_{X,n}$ and $\sigma^2_{X,n}$ in Figure \ref{fig:ouMomBias}.

\begin{figure}
\begin{center}
\includegraphics[width=0.8\textwidth]{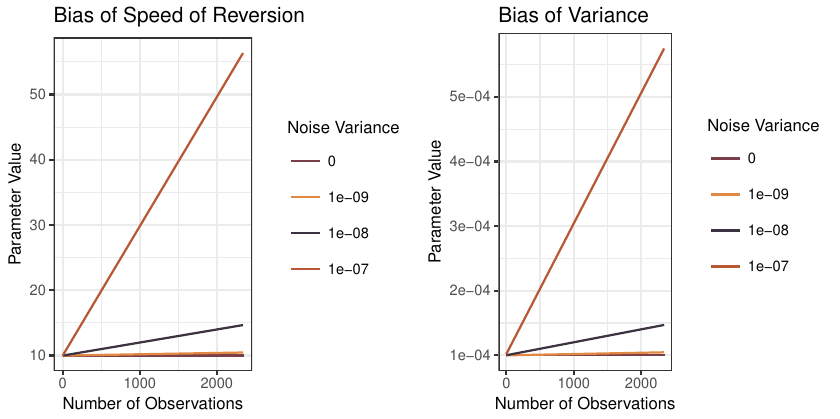}
\caption{The bias of functions $\tau_{X,n}$ and $\sigma^2_{X,n}$ with parameters $\mu=1$, $\tau=10$, $\sigma^2=10^{-4}$ and various values of $\omega^2$.}
\label{fig:ouMomBias}
\end{center}
\end{figure}

\section{Maximum Likelihood Method Without the Noise}
\label{app:likelihood}

In the case of the Ornstein--Uhlenbeck process without the noise, the maximum likelihood estimates are obtained by maximizing the logarithmic likelihood function given by
\begin{equation}
L(\mu, \tau, \sigma^2) = \sum_{i=1}^n \log f_{P_{T_i}} \left( p_{T_i} \mid P_{T_{i-1}} = p_{T_{i-1}} \right),
\end{equation}
where $f_{P_{T_i}} \left( p_{T_i} \mid P_{T_{i-1}} = p_{T_{i-1}} \right)$ is the conditional density function of the observations. According to equation \eqref{eq:ouMomentsCon}, it is the conditional density function of the normal distribution
\begin{equation}
\begin{aligned}
f_{P_{T_i}} \left( p_{T_i} \mid P_{T_{i-1}} = p_{T_{i-1}} \right) &= \frac{1}{\sqrt{2 \pi \mathrm{var}[P_{T_i} \mid P_{T_{i-1}} = p_{T_{i-1}}]}} \\
& \quad \times \exp \left\{-\frac{\left( p_{T_i} - \mathrm{E}[P_{T_i} \mid P_{T_{i-1}} = p_{T_{i-1}}] \right)^2}{2\mathrm{var}[P_{T_i} \mid P_{T_{i-1}} = p_{T_{i-1}}]} \right\},
\end{aligned}
\end{equation}
with conditional moments
\begin{equation}
\begin{aligned}
\mathrm{E} [ P_{T_i} \mid P_{T_{i-1}} = p_{T_{i-1}} ] &= p_{T_{i-1}} e^{-\tau \left(T_i - T_{i-1} \right) } + \mu \left( 1 - e^{-\tau \left(T_i - T_{i-1} \right) } \right) , \\
\mathrm{var} [ P_{T_i} \mid P_{T_{i-1}} = p_{T_{i-1}} ] &= \frac{\sigma^2}{2 \tau} \left( 1 - e^{-2 \tau \left(T_i - T_{i-1} \right) } \right). \\
\end{aligned}
\end{equation}
The logarithmic likelihood function can be simplified to
\begin{equation}
\begin{aligned}
L(\mu, \tau, \sigma^2) &= -\frac{1}{2} \sum_{i=1}^n \log \left( 2 \pi \mathrm{var}[P_{T_i} \mid P_{T_{i-1}} = p_{T_{i-1}} ] \right) \\
& \quad - \frac{1}{2} \sum_{i=1}^n \frac{\left( p_{T_i} - \mathrm{E}[P_{T_i} \mid P_{T_{i-1}} = p_{T_{i-1}} ] \right)^2}{\mathrm{var}[P_{T_i} \mid P_{T_{i-1}} = p_{T_{i-1}} ]}.
\end{aligned}
\end{equation}
The estimates are then given by
\begin{equation}
(\hat{\mu}, \hat{\tau}, \hat{\sigma}^2)' = \underset{\mu, \tau, \sigma^2}{\mathrm{arg} \max} \  L(\mu, \tau, \sigma^2) \quad \text{s. t.} \quad \sigma^2 \geq 0.
\end{equation}

\section{AR Reparametrization Without the Noise}
\label{app:ar}

When the noise is not present, the discrete process $P_{T_i}$ can be reparametrized as an AR(1) process. Using \eqref{eq:ouSolution}, the process $P_{T_i}$ can be rewritten as
\begin{equation}
\label{eq:ouArFirst}
P_{T_i} = P_{T_{i-1}} e^{-\tau \Delta} + \mu (1 - e^{-\tau \Delta}) + \sigma \int_{T_{i-1}}^{T_i} e^{-\tau (\Delta - s)} \mathrm{d} W_s.
\end{equation}
We denote
\begin{equation}
\label{eq:ouArReparDet}
\begin{aligned}
\alpha &= \mu (1 - e^{-\tau \Delta}), \\
\varphi &= e^{-\tau \Delta}. \\
\end{aligned}
\end{equation}
We further denote
\begin{equation}
\label{eq:ouArReparStoch}
V_i = \sigma \int_{T_{i-1}}^{T_i} e^{-\tau (\Delta - s)} \mathrm{d} W_s. \\
\end{equation}
From equation \eqref{eq:ouMomentsCon} we have that the random variable $V_i$ is normally distributed with variance
\begin{equation}
\label{eq:ouArReparMom}
\gamma^2 = \textrm{var} [V_i] = \frac{\sigma^2}{2 \tau} \left( 1 - e^{-2 \tau \Delta} \right). \\
\end{equation}
The random variable $V_i$ is independent from $P_{T_{i-1}}$. Using \eqref{eq:ouArReparDet} and \eqref{eq:ouArReparMom}, we can reparametrize the process \eqref{eq:ouArFirst} as an AR(1) process
\begin{equation}
\label{eq:ouArSecond}
P_{T_i} = \alpha + \varphi P_{T_{i-1}} + V_i, \qquad V_i \stackrel{i.i.d.}{\sim} \textrm{N}(0, \gamma^2).
\end{equation}
We can estimate parameters $\alpha$, $\varphi$ and $\gamma^2$ by any suitable method. Finally, by solving equations
\begin{equation}
\label{eq:ouArSystem}
\begin{aligned}
\hat{\alpha} &= \hat{\mu} (1 - e^{-\hat{\tau} \Delta}), \\
\hat{\varphi} &=  e^{-\hat{\tau} \Delta}, \\
\hat{\gamma}^2 &= \frac{\hat{\sigma}^2}{2 \hat{\tau}} \left( 1 - e^{-2 \hat{\tau} \Delta} \right), \\
\end{aligned}
\end{equation}
we get estimates
\begin{equation}
\begin{aligned}
\hat{\mu} &= \frac{\hat{\alpha}}{1 - \hat{\varphi}}, \\
\hat{\tau} &= - \frac{1}{\Delta} \log \hat{\varphi}, \\
\hat{\sigma}^2 &= -2 \frac{1}{\Delta} \frac{\hat{\gamma}^2 }{1 -\hat{\varphi}^2} \log \hat{\varphi}. \\
\end{aligned}
\end{equation}

\section{Data Preprocessing}
\label{app:data}

Careful data cleaning is one of the most important aspects of high-frequency data analysis. We utilize the standard data cleaning procedure for NYSE TAQ database of \cite{Barndorff-Nielsen2009} with some slight modifications. Our procedure consists of the following steps. 

\begin{enumerate}
\item \emph{Retain entries originating from a single exchange. Delete other entries.} This step corresponds to P3 rule of \cite{Barndorff-Nielsen2009}. 
\item \emph{Delete all trades with a timestamp outside the window when the exchange is open.} The normal trading hours of the NYSE exchanges are from 9:30 am to 4:00 pm in the eastern time zone. This step corresponds to P1 rule of \cite{Barndorff-Nielsen2009}.
\item \emph{Delete entries with corrected trades.} For the NYSE TAQ database, corrected trades are denoted by the correction indicator 'CORR' other than 0. This step removes trades that were corrected, changed, or signified as cancel or error and corresponds to T1 rule of \cite{Barndorff-Nielsen2009}.
\item \emph{Delete entries with abnormal trades.} For the NYSE TAQ database, abnormal trades are denoted by the sale condition 'COND' having a letter code, except for 'E', 'F' and 'I'. This step rules out data points that the NYSE TAQ database is flagging up as a problem and corresponds to T2 rule of \cite{Barndorff-Nielsen2009}.
\item \emph{Delete entries which are identified as preferred or warrants.} For the NYSE TAQ database, all trades with the non-empty SUFFIX indicator should be deleted.
\item \emph{Merge entries with the same timestamp.} Merging itself is done using the median price. This step corresponds to T3 rule of \cite{Barndorff-Nielsen2009}. Merging simultaneous entries is quite common in the literature yet controversial as it leads to the largest deletion of data. \cite{Barndorff-Nielsen2009} argue that this onrule seems inevitable. However, there are few recent studies omitting this rule such as \cite{Liu2018} who estimate integrated variance by the pre-averaging estimator using data with multiple observations at the same time and \cite{Blasques2024a} who directly model zero durations in the zero-inflated autoregressive conditional duration model. Nevertheless, in our study, we resort to merging simultaneous entries for simplicity.
\item \emph{Delete entries with the price equal to zero.} This step removes obvious errors in the dataset and corresponds to P2 rule of \cite{Barndorff-Nielsen2009}.
\item \emph{Delete entries for which the price deviated by more than 10 mean absolute deviations from a rolling centred median of 50 observations.} The observation under consideration is excluded in the rolling centered median. This step corresponds to Q4 rule of \cite{Barndorff-Nielsen2009}.
\end{enumerate}

After data cleaning, the parameters of the Ornstein--Uhlenbeck process are estimated. During the estimation, we face the following issue concerning with distribution assumptions. We assume the Ornstein--Uhlenbeck process based on the normal distribution. This is quite restrictive assumption as financial data often exhibit heavy tails and the presence of jumps. Although somewhat rare, large jumps can cause problems for the estimators based on the maximum likelihood. A large jump over short period of time is not consistent with the assumed volatility process which is proportional to the time period and the maximum likelihood estimator attributes this jump to the noise component. This results in zero variance of the Ornstein--Uhlenbeck process $\sigma^2$ and overestimation of the noise variance $\omega^2$. To avoid such problems, we consider large jumps to be outliers and remove them from data for the estimation purposes. We remove 1 \% of all observations with the lowest log likelihood at initial parameter values. In the subsequent analysis, removed observations are again included.

\section{Supplementary Empirical Sample}
\label{app:new}

The analysis in Section \ref{sec:oil} is based on data from January 2, 2015, to June 29, 2018, consisting of 880 trading days. In this appendix, we perform the analysis on a newer data sample from March 1, 2022, to December 29, 2023, consisting of 462 trading days. This is a period of turmoil, as oil prices and stock volatility spiked after Russia invaded Ukraine in late February 2022.

For this period, we use a different data source -- Refinitiv Eikon. The data are in a slightly different format; notably, the timestamps are recorded with only millisecond precision. Consequently, merging entries with the same timestamp, as described in Appendix \ref{app:data}, reduces the dataset to a greater extent.

The differences between the TICK-MLE and TICK-MLE-NR estimators are less pronounced but still notable. On average, the speed of reversion $\tau$ is 2.37 times higher, and the standard error $\sigma$ is 1.62 times higher (with the variance $\sigma^2$ being 2.56 times higher) when estimated using the TICK-MLE method. This is caused by the aggregation of data over time, as the timestamps are recorded only with millisecond precision, which reduces the impact of market microstructure noise.

Figure \ref{fig:ouStrategyProfitNew} shows the profitability of the strategy during this period. We can see that neither the TICK-MLE nor the TICK-MLE-NR estimator leads to profits for any combination of minimum mean and maximum variance levels. Higher levels of the minimum mean tend to perform better, as fewer trades are made. The best option would be not to trade at all, given the level of transaction costs ($c = 0.0015$ per round-trip pair trade). When at least a small value for the mean is required, the TICK-MLE-NR estimator results in smaller losses compared to the TICK-MLE estimator, which overestimates both the variance and the speed of reversion of the Ornstein–Uhlenbeck process. In this way, the behavior is similar to that in Figure \ref{fig:ouStrategyProfit}, where some profits are achieved.

\begin{figure}
\begin{center}
\includegraphics[width=0.8\textwidth]{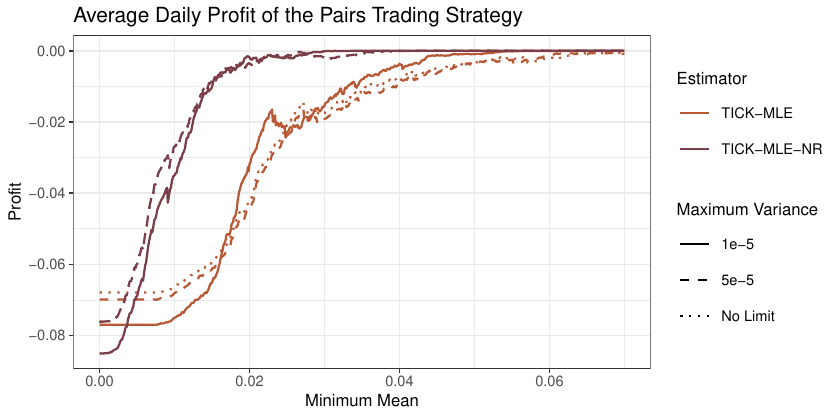}
\caption{Dependence of the daily profit on the maximum variance $\eta$ and minimum mean $\zeta$ for the noise-sensitive and noise-robust estimators from March 1, 2022 to December 29, 2023.}
\label{fig:ouStrategyProfitNew}
\end{center}
\end{figure}


\end{document}